\RequirePackage{fix-cm}
\documentclass[smallextended]{svjour3}       % onecolumn (second format)
\smartqed  % flush right qed marks, e.g. at end of proof
\usepackage{tikz}
\usetikzlibrary{arrows.meta}

\usepackage{amsmath, amssymb, amsfonts}

\usepackage{graphicx}

\usepackage{ntheorem}

\makeatletter
\renewtheoremstyle{plain}% Adds automatic line break, if heading is too long
  {\item{\theorem@headerfont ##1\ ##2\theorem@separator}~}
  {\item{\theorem@headerfont ##1\ ##2\ (##3)\theorem@separator}~}
\makeatother

{\theoremheaderfont{\upshape\bfseries}
 \theorembodyfont{\normalfont\slshape}
\newtheorem{alg}{Algorithm}
\newtheorem{cor}{Corallary}
\newtheorem{thm}[theorem]{Theorem}
}

\begin{document}

\newcommand{\argmin}{\text{argmin}}
\newcommand{\eps}{\varepsilon}
\newcommand{\To}{\longrightarrow}
\newcommand{\h}{\mathcal{H}}
\newcommand{\s}{\mathcal{S}}
\newcommand{\A}{\mathcal{A}}
\newcommand{\J}{\mathcal{J}}
\newcommand{\M}{\mathcal{M}}
\newcommand{\W}{\mathcal{W}}
\newcommand{\BOP}{\mathbf{B}}
\newcommand{\BH}{\mathbf{B}(\mathcal{H})}
\newcommand{\KH}{\mathcal{K}(\mathcal{H})}
\newcommand{\Real}{\mathbb{R}}
\newcommand{\Complex}{\mathbb{C}}
\newcommand{\Field}{\mathbb{F}}
\newcommand{\RPlus}{\Real^{+}}
\newcommand{\Polar}{\mathcal{P}_{\s}}
\newcommand{\Poly}{\mathcal{P}(E)}
\newcommand{\EssD}{\mathcal{D}}
\newcommand{\Lom}{\mathcal{L}}
\newcommand{\States}{\mathcal{T}}
\newcommand{\abs}[1]{\left\vert#1\right\vert}
\newcommand{\set}[1]{\left\{#1\right\}}
\newcommand{\seq}[1]{\left<#1\right>}
\newcommand{\essnorm}[1]{\norm{#1}_{\ess}}
\newcommand{\bu}{{\bf u}}
\newcommand{\obu}{{\overline{{\bf u}}}}
\newcommand{\obv}{{\overline{\bf v}}}
\newcommand{\bv}{{\bf v}}
\newcommand{\bw}{{\bf w}}
\newcommand{\bs}{{\bf s}}
\newcommand{\bg}{{\bf g}}
\newcommand{\bq}{{\bf q}}
\newcommand{\bx}{{\bf x}}
\newcommand{\bp}{{\bf p}}
\newcommand{\br}{{\bf r}}
\newcommand{\by}{{\bf y}}
\newcommand{\bE}{{\bf E}}
\newcommand{\bn}{{\bf n}}
\newcommand{\be}{{\bf e}}
\newcommand{\bU}{{\bf U}}
\newcommand{\bP}{{\bf P}}
\newcommand{\bo}{{\bf 0}}
\newcommand{\bff}{{\bf f}}
\newcommand{\bphi}{{\boldsymbol \phi}}
\newcommand{\obphi}{{\overline{\boldsymbol \phi}}}
\newcommand{\bvarphi}{{\boldsymbol \varphi}}
\newcommand{\bX}{{\bf X}}
\newcommand{\bV}{{\bf V}}
\newcommand{\bH}{{\bf H}}
\newcommand{\bQ}{{\bf Q}}
\newcommand{\bomega}{{\boldsymbol \omega}}
\newcommand{\bpsi}{{\boldsymbol \psi}}
\newcommand{\bfeta}{{\boldsymbol \eta}}
\newcommand{\bchi}{{\boldsymbol \chi}}
\newcommand{\bzeta}{{\boldsymbol \zeta}}
\newcommand{\blambda}{{\boldsymbol \lambda}}

\title{An Online, Dynamic Amplitude-Correcting Gradient Estimation Technique to Align X-ray Focusing Optics % \at Grants or other notes
%about the article that should go on the front page should be
%placed here. General acknowledgments should be placed at the end of the article.}
}

%\titlerunning{Short form of title}        % if too long for running head

\author{Sean Breckling \and
Leora E. Dresselhaus-Marais \and
Bernard Kozioziemski \and 
Michael C. Brennan \and
Malena Espa\~{n}ol \and
Ryan Coffee \and
Sunam Kim \and
Sangsoo Kim \and
Daewoong Nam \and
Arnulfo Gonzalez \and
Margaret Lund \and
Jesse Adams \and
Jordan Pillow \and
Eric Machorro \and
Daniel Champion \and
Kevin Joyce \and
Ajana\'{e} Williams \and
Marylesa Howard}

\authorrunning{Breckling, et al} % if too long for running head
\titlerunning{An Online, Dynamic Amplitude-Correcting Gradient Estimation Technique}

\institute{Sean Breckling \at
Nevada National Security Site
\email{brecklsr@nv.doe.gov}           %  \\
\and
Leora E. Dresselhaus-Marais \at
Stanford University
\email{leoradm@stanford.edu}
\and Bernard Kozioziemski \at
Lawrence Livermore National Laboratory, \email{kozioziemski1@llnl.gov}}

\date{Received: date / Accepted: date}
% The correct dates will be entered by the editor

\maketitle

\begin{abstract}
High-brightness X-ray pulses, as generated at synchrotrons and X-ray free electron lasers (XFEL), are used in a variety of scientific experiments. Many experimental testbeds require optical equipment, e.g Compound Refractive Lenses (CRLs), to be precisely aligned and focused. The lateral alignment of CRLs to a beamline requires precise positioning along four axes: two translational, and the two rotational. At a synchrotron, alignment is often accomplished manually. However, XFEL beamlines present a beam brightness that fluctuates in time, making manual alignment a time-consuming endeavor. Automation using classic stochastic methods often fail, given the errant gradient estimates. We present an online correction based on the combination of a generalized finite difference stencil and a time-dependent sampling pattern. Error expectation is analyzed, and efficacy is demonstrated. We provide a proof of concept by laterally aligning optics on a simulated XFEL beamline, generated from data collected at Pohang Accelerator Laboratory XFEL, and the Advanced Photon Source at Argonne National Laboratory.
%High-brightness X-ray pulses, as generated at synchrotrons and X-ray free electron lasers (XFELs), are used in a variety of scientific experiments. At these facilities, measurements often require optical equipment to be precisely aligned and focused on the beam. Compound refractive lenses (CRLs) are an example of those beamline optics. The alignment of CRLs to a beamline requires precise positioning along four axes: two translational (perpendicular to the beam), and the two rotational axes about them. At synchrotrons, alignment is usually accomplished manually. However, the intensity of the beam fluctuates unpredictably at XFELs, compounding the difficulty. Automation using simplex or classic stochastic descent often fails, given the errant online gradient estimates. Herein we present a dynamic-amplitude correction to the usual gradient based on the combination of a generalized finite difference stencil and a time-dependent sampling pattern. Intensity is recorded periodically, then used to normalize numerical derivatives against fluctuations. Error expectation is analyzed, and efficacy is demonstrated on classic benchmarks. We demonstrate proof of concept by laterally aligning optics on a simulated beamline using recorded beam intensity from an experiment at the PAL-XFEL.
% \PACS{PACS code1 \and PACS code2 \and more}
% \subclass{MSC code1 \and MSC code2 \and more}
\end{abstract}
\section{Introduction}
\label{sec1}
Compound refractive lens assemblies (CRLs) are frequently used as objective lenses in X-ray microscopes \cite{lengeler99}, or as upstream condensers \cite{schroer05,vaughan11}. Both applications are extremely sensitive to lateral misalignment. The full focusing procedure of a CRL requires five degrees of freedom. Independent translations along the $x,$ and $y$ axes, in addition to two rotations $r_x$ and $r_y$ about those axes, produce a lateral alignment. The fifth degree of freedom is a translation along the $z$ axis, which locates the classical position of the focus (see Figure~\ref{hugh_image}). 

Our principal motivation for this work is the task of laterally aligning compound refractive lens assemblies (CRLs) at  X-ray free electron laser facilities (XFELs). XFELs are a new class of X-ray sources that produce the shortest duration and brightest X-ray pulses currently attainable, paving the way for experiments that were previously not possible \cite{Yabashi2017}. A new approach to the alignment of CRLs is necessary, in large part due to the novel amplification process to generate X-ray pulses. At XFEL facilities,  Self-Amplification of Spontaneous Emission (SASE) causes the beam position, spatial mode, propagation direction (pointing), and intensity to fluctuate stochastically \cite{Emma2010,Schneidmiller2016}. The proper lateral alignment of focusing optics is crucial to produce the highest resolution and smallest focal spots, as required for many modern X-ray experiments.

\begin{figure}[h]
    \centering
    \includegraphics[width=5cm]{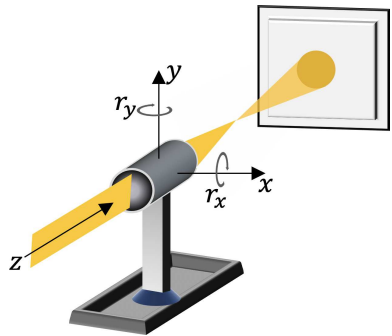}
    \caption{This example diagram depicts a CRL-based imaging configuration along an optical axis $z$. A CRL is aligned to the optical axis by four independent motors. Two control translations along the perpendicular $x$ and $y$ axes, and two control rotations about the $x$ and $y$ axes.}
    \label{hugh_image}
\end{figure}

The orientation of the CRL is controlled by four motorized stages: two that translate the optic along the $x$ and $y$ axes, and two that rotate about them. To align this optic, a detector is placed behind the exit surface of the CRL to measure the transmitted X-ray beam. Typically, these sensors are either charge-coupled device (CCD) cameras or different implementations of a photodiode (e.g. ion chamber). The beam-line scientist's task is to laterally align the CRL by maximizing X-ray transmitted light that reaches the sensor as a function of the four positions. 

In more formal language, let $f\colon \mathbb{R}^4 \to \mathbb{R}$ denote an idealized model of a noiseless, steady (in space and amplitude) X-ray transmission through a CRL as a function of the orientation. In \cite{simons17}, Simons et al. demonstrated that for a convex region $\Omega \subset \mathbb{R}^4$, that transmission $f$ can be modelled with a Gaussian distribution. In this idealized case, the alignment procedure reduces to the trivially convex optimization problem
\begin{equation*}
\min_{\bx \in \Omega} -f(\bx).
\end{equation*}

Given the simplicity of the problem, it is common at scientific beam-lines with more-stable amplitudes to establish lateral alignment through a simple manual process. A recent study automated this task by using a modified stochastic simplex method on the lateral alignment of CRL assemblies at synchrotron facilities \cite{Breckling2021}. In lieu of automating, the usual approach is to perform a rough initial alignment, then select two of the four dimensions of $\Omega$ and perform a raster scan of the transmission, logging a detector's response at each particular orientation. The ``best" position from that 2D scan is selected, and the micro controllers are driven to that position. The alternate dimensions are then selected, and the procedure repeats until alignment is satisfactory. 

For many scientific beam-lines, this dead-reckoning approach is sufficient. Unfortunately, the SASE process for generating X-rays at XFEL facilities introduces an unpredictable time-dependent intensity drift, as well as stochastic perturbations of the beam's propagation axes. This, of course, is in addition to the usual sources of measurement noise. These complications prevent the reliable success of a direct implementation of the simplex-based approach seen in \cite{Breckling2021}. Given that it is only possible to record X-ray transmission for a single orientation at a single moment in time, an orderly raster-like scan of the transmission at an XFEL facility is not likely to see a distribution that strongly agrees with the Gaussian model developed in \cite{simons17}. As a result, it remains the common practice to rely heavily on the intuition of the beamline scientist to interpret such scans, substantially extending the time required to produce an acceptable initial lateral alignment, and realignment. Given that time is an extremely limited resource at XFELs, an alternative technique to quickly and reliably expedite this procedure is sought. 

In this paper we propose a technique to estimate the gradient of the transmission function that accounts for both time-dependent amplitude fluctuations, and instrumentation noise. If successful, such a gradient could be utilized in a classic steepest descent algorithm. Given that stochastic descent-based approaches have been successful in automating similar optical alignment and focusing tasks, including the control of directed energy sources \cite{belen2007laboratory}, aligning line-of-sight communication arrays \cite{Raj2010}, and the alignment of two-mirror telescopes \cite{Li20}, we suspect that these corrections will allow for expedient and accurate alignments in our application. 

Let $t \in \mathbb{R}^+$ represent time,  $T\colon \mathbb{R}^+ \to \mathbb{R}$ be an arbitrarily smooth function which denotes the intensity of the beam over time, $\varsigma$ be the aggregate of all additive stochastic noise, and $\Theta$ denote stochastic perturbations to the beam's orientation. We then formulate our estimate of the transmission function as $G(\bx,t)=-T(t)f(\bx+\Theta) + \varsigma$. Our alignment procedure then looks like the optimization problem
\begin{equation}\label{exp_prob}
\min_{\bx \in \Omega} E\left[G(\bx,t)\right], \text{ for all } t \geq 0, %\label{mainprob1}
\end{equation}
where $E$ is the expected value. While this problem does admit an optimal solution, common stochastic steepest-descent methods are not amenable to finding it without directly addressing the amplitude fluctuations \cite{Spall}. 

We propose an approach to this problem that substitutes the usual finite difference method with one which corrects for the non-steady amplitude. The method systematically intertwines the usual spatial samples for the gradient with additional samples from a fixed central location. We demonstrate through error asymptotics and numerical benchmarks that these additional samples can, when collected at sufficient rate, can sufficiently account for amplitude changes in intensity over time. Thus, an amplitude-corrected gradient, when paired with a standard stochastic descent algorithm, becomes well-suited for minimization problems like \eqref{exp_prob}.

The remainder of the paper is organized as follows: We formally introduce the amplitude-correcting scheme in Section~\ref{DiffScheme}, along with notation, and asymptotic error estimates. We provide two numerical benchmarks in Section~\ref{NumericalExperiments}. There, we first develop asymptotic error estimates for stochastic gradient descent (SGD) schemes using the amplitude-correction, along with a demonstration of the resulting convergence rates. We then demonstrate the efficacy of our amplitude-correcting gradient on a modified version of the Rosenbrock valley benchmark. Section~\ref{XFEL} outlines how our method shows promise in automating the lateral alignment of CRLs at X-ray experimental facilities. There, we provide a proof-of-concept implementation of our full optimization scheme against a synthetic cost function modelled to behave appreciably similar to one used at a genuine XFEL facility. Finally, we provide remarks in summary in Section~\ref{conclusions}.

\section{Constructing the Amplitude-Correcting Differencing Scheme} \label{DiffScheme}
Given a function $f\colon \mathbb{R}^n \to \mathbb{R}$, we are primarily concerned with computing estimates of $\nabla f$. To this end, we assume a high degree of smoothness, i.e., $f$ is sufficiently G\^{a}teux and Fr\'{e}chet differentiable to satisfy the necessary conditions of our estimates to follow. The gradient of a function at a particular point $\bx_c \in \mathbb{R}^n$ is typically estimated by sampling that function $\mathcal{O}(N)$ $(N \in \mathbb{N}, N > n)$ times in a local region around $\bx_c$. We consider an $n-$dimensional ball of radius $\delta>0$ centered at $\bx_c$, denoted $\mathcal{B}_\delta(\bx_c)$, and define $\Omega$ to be an open, connected, bounded set containing $\mathcal{B}_\delta(\bx_c)$ within the interior. 

For our application, $n=4$, given the degrees of freedom for lateral alignment. Additionally, the function $f$ can only be evaluated at one particular position $\bx \in \mathbb{R}^n$ at a time. Sampling another position $\bx' \in \mathbb{R}^n$ requires a discrete amount of time $h>0$ to elapse. Given a particular starting time $t_0 > 0$, we denote the interval of time required to compute a gradient using our technique defined below to be 
\begin{equation*}
    \mathcal{T}_{h,N} =: [t_0,t_0 + (4N + 1)h].
\end{equation*}
Though in the interest of brevity, we may refer to $\mathcal{T}_{h,N}$ as simply $\mathcal{T}$. Finally, we assume that our amplitude function $T\colon \mathbb{R}^+ \to \mathbb{R}$ is at least four-times differentiable, i.e., $T \in \mathcal{C}^4(\mathcal{\mathcal{T})}.$

%{\color{purple}[This designates how many times you can take a derivative of these]} 
Let $E$ and $V$ denote the expectation and variance of a time series over $\mathcal{T}$. Our source of additive noise is assumed to be normal and i.i.d. such that $E(\varsigma)$ = 0 and $V(\varsigma) = \sigma^2$. Our smooth and additive noise-corrupted functions are written:
\begin{eqnarray*}
	F(\bx,t) &=& T(t)f(\bx), \\
	G(\bx,t) &=& F(\bx,t) + \varsigma(t).
\end{eqnarray*}

To organize our scheme, we arrange our sample indices serially in terms of the position in $\mathcal{B}_\delta(\bx_c)$ and time $t_k \in \mathcal{T}$. Let $\be$ be an arbitrary unit vector in $\mathbb{R}^n$. For the noise-free case, we write:
\begin{center}
	\begin{tabular}{rclcl}
		$F(\bx_c,t_k)$                       & = & $F_k^c$                 & = & $T_k f^c$, \\
		$F(\bx_c \pm \delta \be, t_k \pm h)$ & = & $F_{k\pm1}^{\be^{\pm}} $& = & $T_{k\pm1} f^{\be^{\pm}}$.
	\end{tabular}
\end{center}
\noindent Similarly, our noise corrupted case is written:
\begin{center}
	\begin{tabular}{rclcl}
		$G(\bx_c,t_k)$                       & = & $G_k^c$                 & = & $F_k^c + \varsigma_k$, \\
		$G(\bx_c \pm \delta \be, t_k \pm h)$ & = & $G_{k\pm1}^{\be^{\pm}} $& = & $F_{k\pm1}^{\be^{\pm}} + \varsigma_{k\pm1}$,
	\end{tabular}
\end{center}
where $\be$ is a unit vector in the selected direction.

We use the over-bar shorthand to denote time-averaged terms, e.g.,
\begin{equation*}
\bar{F}_k^c = \frac{F_{k-1}^c + F_{k+1}^c}{2}.
\end{equation*}
We make use of the usual norm notation, i.e., $|| \cdot ||_2$ denotes an $L^2$ norm; though the subscript is dropped in the context of Euclidean vectors. When discussing discretized approximations to the usual gradient operator $\nabla$, we use $\nabla_\delta$ to denote the uncorrected differencing scheme provided in Definition \ref{basic_grad}, and $\nabla_{\delta,h}$ for the amplitude-correcting gradient estimate developed further below. Directional derivative operators and their approximations are then written as $(\be \cdot \nabla)$, $(\be \cdot \nabla_\delta)$, and $(\be \cdot \nabla_{\delta,h})$ respectively. 

\begin{definition}[A Linear Regression-Based Gradient Estimate]\label{basic_grad} Let $\delta > 0$, and $\Omega \subset \mathbb{R}^n$ contain the open ball $\mathcal{B}_{\delta}(\bx_c)$. Further, let the points $\lbrace \bx_i \rbrace_{i=1}^{N}$ be a collection of $N$ unique points on the surface of the ball $\mathcal{B}_{\delta}(\bx_c)$ such that $N>n$. For a given function $f :\Omega \rightarrow \mathbb{R}$, sample each point on the ball, collecting each sample in the vector ${\bf F} = \lbrace f_i \rbrace_{i=1}^{N}$.  Use the matrix $\bX = \lbrace 1, \bx_i \rbrace_{i=1}^{N}$ and corresponding samples ${\bf F}$ to assemble the linear regression problem 
\begin{equation*}
     \bfeta  = (\bX^T \bX)^{-1} \bX^T {\bf F}.
\end{equation*}
The solution $\bfeta = \lbrace \eta_i \rbrace_{i=1}^{n+1}$ determines the gradient estimate
\begin{equation*}
    \nabla f(\bx_c) \approx \nabla_\delta f(\bx_c) =: \lbrace \eta_i \rbrace_{i=2}^{n+1}.
\end{equation*}
\end{definition}

\subsection{The Differencing Scheme}
\label{sec:diffscheme}
The definition below is assembled similarly to that seen in Definition \ref{basic_grad}, but coordinates all sampling according to a uniformly-discretized time series. A example diagram is provided in Figure \ref{ball}. If the sampling distance $\delta > 0$ remains uniform, it is assumed that the time required to visit each point within the sequence is uniform. While this isn't a necessary limitation in practice, this assumption simplifies the analysis provded in \ref{AccuracyEst}.

\begin{figure} 
	\begin{center}
		{\bf An Example 6-Point Stencil in 2D} \\
		\begin{tikzpicture}[fill = white]
		\draw[blue!20,thick,dashed] (0,0) circle (2cm);
		\path (0,0)       node(a) [circle, draw, fill] {$\bx_c$}
		(2.0,0.0)   node(b) [circle, draw, fill] {$\bx_1$}
		(1, 1.7320) node(c) [circle, draw, fill] {$\bx_3$}
		(-1,1.7320) node(d) [circle, draw, fill] {$\bx_5$}
		(-2.0,0)    node(e) [circle, draw, fill] {$\bx_2$}
		(-1,-1.7320)node(f) [circle, draw, fill] {$\bx_4$}
		( 1,-1.7320)node(g) [circle, draw, fill] {$\bx_6$};
		
		\draw[blue!100,thick,-{Straight Barb[left]}]     (node cs:name=a, angle=8    )  --  (node cs:name=b, angle=180-8);
		\draw[blue!100,  thick,-{Straight Barb[left]}]     (node cs:name=b, angle=180+8)  --  (node cs:name=a, angle=-8   ); 
		
		\draw[purple!100,thick,-{Straight Barb[left]}]     (node cs:name=a, angle=60+8    )  --  (node cs:name=c, angle=180+60-8);
		\draw[purple!100,  thick,-{Straight Barb[left]}]     (node cs:name=c, angle=180+60+8)  --  (node cs:name=a, angle=60-8   ); 			
		
		\draw[green!100,thick,-{Straight Barb[left]}]     (node cs:name=a, angle=120+8    )  --  (node cs:name=d, angle=180+120-8);
		\draw[green!100,  thick,-{Straight Barb[left]}]     (node cs:name=d, angle=180+120+8)  --  (node cs:name=a, angle=120-8   ); 	
		
		\draw[blue!100,thick,-{Straight Barb[left]}]     (node cs:name=a, angle=180+8    )  --  (node cs:name=e, angle=180+180-8);
		\draw[blue!100,  thick,-{Straight Barb[left]}]     (node cs:name=e, angle=180+180+8)  --  (node cs:name=a, angle=180-8   ); 			
		
		\draw[purple!100,thick,-{Straight Barb[left]}]     (node cs:name=a, angle=240+8    )  --  (node cs:name=f, angle=180+240-8);
		\draw[purple!100,  thick,-{Straight Barb[left]}]     (node cs:name=f, angle=180+240+8)  --  (node cs:name=a, angle=240-8   ); 			
		
		\draw[green!100,thick,-{Straight Barb[left]}]     (node cs:name=a, angle=300+8    )  --  (node cs:name=g, angle=180+300-8);
		\draw[green!100,  thick,-{Straight Barb[left]}]     (node cs:name=g, angle=180+300+8)  --  (node cs:name=a, angle=300-8   ); 			
		\end{tikzpicture}
	\end{center}
	\caption{\label{ball} This depicts an example six-point ($N=3$) sampling stencil  for a two-dimensional search space. The procedure requires a total of 13 samples. Begin by sampling at $\bx_c.$ Next, sample at $\bx_1$, then return and sample $\bx_c$. Repeat this process sequentially for the remaining $\bx_i$. }
\end{figure}
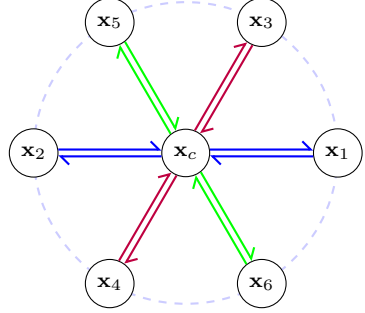

\begin{definition}[Amplitude-Correcting Gradient Estimate]\label{def_grad} Let $\delta, h > 0,$ and $\Omega \subset \mathbb{R}^n$ contain the open ball $\mathcal{B}_{\delta}(\bx_c)$. Let the points $\lbrace \bx_i \rbrace_{i=1}^{N}$ be a collection of $N$ unique points on the surface of the ball $\mathcal{B}_{\delta}(\bx_c)$, along with the $N$ corresponding antipodal points $\lbrace \bx_i' \rbrace_{i=1}^N$, such that $N>n$. Let $\mathcal{T}_{h,N}$ be the uniform discretization of the time interval $\mathcal{T}.$ For the function $f:\Omega \rightarrow \mathbb{R}$, the amplitude function $T:\mathcal{T} \rightarrow \mathbb{R}$, and additive noise $\varsigma: \mathcal{T} \rightarrow \mathbb{R}$, we write the given function $G : \Omega \times \mathcal{T} \rightarrow \mathbb{R}$ such that $G(\bx,t) = T(t)f(\bx) + \varsigma(t)$. Let $\mu_T = E\left[T\left(\mathcal{T}_{h,N} \right) \right]$, and $\be_k^+$ be the unit vector in the direction $\bx_k - \bx_c.$ Index our uniformly discretized time-steps as $k=1,\ldots,4N+1$. Sample $\bx_c$ at odd values of $k,$  i.e.  $k = 2k'-1$, collecting each sample as $G_{k}^c.$ Time-average these values gives $\bar{G}_{2k'}^c$.  For even values of $k,$ i.e. $k= 2k'$, alternate sampling $\bx_{k'}$ and its antipodal counterpart $\bx_{k'}'$, collecting each sample as $G_{k}^{\be_{k}^+}.$  We organize the matrix $\bX$ such that
\begin{equation*}
    \bX = \begin{bmatrix}
   1 & \bx_1  \\
   1 & \bx_1' \\
   \vdots & \vdots \\
   1 & \bx_N \\
   1 & \bx_N'
\end{bmatrix}, 
\end{equation*}
the sample matrix ${\bf G}$ such that
\begin{equation*}
    {\bf G} = \frac{1}{\mu_T}\lbrace G_{2i}^{\be_{2i}^+} - \bar{G}_{2i}^c , \rbrace_{i=1}^{2N}
\end{equation*}and the linear regression problem 
\begin{equation*}
    \bfeta = (\bX^T \bX)^{-1} \bX^T {\bf G}.
\end{equation*}
The solution $\bfeta = \lbrace \eta_i \rbrace_{i=1}^{n+1}$ determines our gradient estimate
\begin{equation*}
    \nabla f(\bx_c) \approx \nabla_{\delta,h}G(\bx_c, \mathcal{T}) =: \lbrace \eta_i \rbrace_{i=2}^{n+1}.
\end{equation*}
\end{definition}

If we use the following short-hand for the standard central-differencing stencil (in spatial coordinates), directional derivatives can be written 
\begin{equation*}
\left(\be_k \cdot \nabla_{\delta}\right)f(\bx_c) = \frac{1}{2\delta} \left(f^{\be_k^+} - f^{\be_k^-}\right).
\end{equation*}
Our convention of selecting antipodal points in sequence allows us to utilize these directional derivative stencils directly. Since each observation of $G(\bx,t)$ results in an independent noise term $\varsigma$, combining like-terms results in
\begin{align*}
	\left(\be_k \cdot \nabla_{\delta,h}\right)G(\bx_c,\mathcal{T})  =  \frac{1}{2\mu_T\delta} & \left[\left(G_{2k}^{\be_k^+} - \bar{G}_{2k}^{c} \right) - \left(G_{2k+2}^{\be_k^-} - \bar{G}_{2k+2}^{c}\right)\right] \\
	 =  \frac{1}{2\mu_T\delta} & \Big[\left(F_{2k}^{\be_k^+} - \bar{F}_{2k}^{c} \right) - \left(F_{2k+2}^{\be_k^-} - \bar{F}_{2k+2}^{c}\right) \\
	 & + \frac{\varsigma_{1,k}}{2} + \varsigma_{2,k} + \varsigma_{3,k} + \frac{\varsigma_{4,k}}{2}\Big].
\end{align*}

\begin{thm}[Error Estimate on Noise-Free Functions]\label{NoiseFreeThm} Let $F(\bx,t) = T(t)f(\bx)$ where $T$ and $f$ are at least $\mathcal{C}^4(\mathcal{T})$ and $\mathcal{C}^3(\Omega)$ respectively. We sample $N$ antipodal pairs such that the resulting sampling is unbiased, and quasi-uniform. For $\bx_c \in \Omega$, and $\delta > 0$ such that $\mathcal{B}_\delta(\bx_c)$ is in the interior of $\Omega$, we let $\be_k$ be the unit vector associated with the $k^{th}$ antipodal pair of points. Further, we let $\mu_T$ be the known expectation of $T(t)$ over $\mathcal{T}$. Selecting $h$ such that $h^3 < \delta$ guarantees that there exists a constant $C^*(\delta, h, N, T,T',T^{(4)},f, \nabla f)  > 0$ such that, 
	\begin{equation*}
	\left|\left|  \nabla f(\bx_c) - \nabla_{\delta,h} F(\bx_c,\mathcal{T})\right|\right| \leq C^* \left(h + \delta^2 \right).
	\end{equation*}
\end{thm}

\noindent A similar result is provided for the case when additive i.i.d. noise is present. 

\begin{thm}[Error Estimate on Noisy Functions]\label{NoisyThm} Let $G(\bx,t) = F(\bx,t) + \varsigma(t)$. Under the same assumptions as Theorem \ref{NoiseFreeThm}, the total contribution of error from stochastic sources can be written
	\begin{eqnarray*}
		{\varepsilon} &:=& \nabla_{\delta,h} G(\bx_c, \mathcal{T}) - \nabla_{\delta,h} F(\bx_c, \mathcal{T}) \\
		& =& \left\lbrace \sum_{k=1}^N \frac{\hat{\be}_i^T \cdot \be_k}{2 \mu_T N \delta} \left[\frac{\varsigma_{1,k}}{2} + \varsigma_{2,k} + \varsigma_{3,k} + \frac{\varsigma_{4,k}}{2}\right]\right\rbrace_{i=1}^n.
	\end{eqnarray*} 
	Then it follows that
	\begin{equation*}
	    E\left[||\varepsilon||\right] \leq 4 \frac{\sigma}{\mu_T \delta} \sqrt{\frac{n}{N}},
	\end{equation*}
	and for $p \in (0,1)$, the probability 
	\begin{equation*}
	    \mathcal{P} \left[ ||\varepsilon|| \leq 4 \frac{\sigma}{\mu_T \delta} + 2 \frac{\sigma}{\mu_T \delta}\sqrt{\frac{\log(1/p)}{N}}\right] \geq 1-p.
	\end{equation*}
\end{thm}

\section{Numerical Demonstrations} \label{NumericalExperiments}
Given that our motivation is to employ the amplitude-correcting gradient in steepest descent methods, our demonstrations will focus on that application. We begin by presenting two accelerated versions of the classic SGD algorithm, differing only by which gradient estimation technique utilized. A full discussion on proper choices for $\alpha$ and $\beta$ can be found in \cite{Nesterov}.

\begin{alg}[Accelerated SGD] \label{Alg2} Choose a suitable initial condition $\bx_0 \in \mathbb{R}^n$, step-size $\alpha_i > 0$, $\alpha_i \rightarrow 0$ as $i \rightarrow \infty$, and $\beta \in [0,1)$. Additionally, choose a radius $\delta > 0$ for the gradient estimator. Indexing our steps with $i=0,1,\ldots$ we proceed such that
	\begin{eqnarray*}
		\by_{i+1} &=& \beta \by_i - \nabla_{\delta}f(\bx_{i+1}), \\
		\bx_{i+1} &=& \bx_i - \alpha_i \by_{i+1}.
	\end{eqnarray*}
\end{alg}

\begin{alg}[Dynamic Amplitude-Corrected Accelerated SGD] \label{Alg1} Choose a suitable initial condition $\bx_0 \in \mathbb{R}^n$, step-size $\alpha_i > 0$, $\alpha_i \rightarrow 0$ as $i \rightarrow \infty$, and $\beta \in [0,1)$. Additionally, prescribe a spatial radius and time-step $\delta, h > 0$ for the gradient estimator. Indexing our steps with $i=0,1,\ldots$ we proceed such that
	\begin{eqnarray*}
		\by_{i+1} &=& \beta \by_i - \nabla_{\delta,h}G(\bx_{i+1},t), \\
		\bx_{i+1} &=& \bx_i - \alpha_i \by_{i+1}.
	\end{eqnarray*}
\end{alg}

In our first demonstration, we seek a direct comparison of the classic SGD algorithm with the amplitude-correcting version. In order for such a comparison to be salient, we consider two functions: Rosenbrock's valley with and without a time-varying amplitude. We then demonstrate, for well-selected parameters, that Algorithm \ref{Alg2}'s performance on the steady-amplitude function qualitatively matches Algorithm \ref{Alg1}'s performance on the non-steady version. When both simulations are successful against minimization problems that are otherwise formulated identically, we can conclude that the amplitude-corrections encoded into the online gradient estimate effectively overcome the variations.

In the second numerical experiment, we show that the error asymptotics provided in Theorems \ref{NoiseFreeThm} and \ref{NoisyThm} can be seen in SGD executions. We cite two theorems that respectively provide sufficient conditions for the convergence of Algorithm \ref{Alg2} with probability 1, and asymptotic error estimates. We then construct a noisy, time-varying function that otherwise adheres to those conditions, then prove that well-selected parameters guarantee Algorithm \ref{Alg1} also converges. This is numerically verified by isolating each source of error to see if the analytic rates match those encountered numerically. 

\subsection{A Quake in Rosenbrock's Valley}\
Rosenbrock's Valley \cite{Rosenbrock} is a polynomial on $\mathbb{R}^2$ defined as
\begin{equation}\label{rosenbrock}
f(x,y) = (1-x^2) + 100(y-x^2)^2.
\end{equation}
This polynomial has a global minimum value of $f(1,1) = 0$, and is locally convex around that point. However, the downward slope along the minimal ridge is quite low in the parabolic valley. It is this feature that made Rosenbrock's Valley a popular benchmark, since many steepest descent algorithms tend to reach the ridge quite quickly, but struggle to reach the optimal answer due to the oscillations spurred from the large values of $|\nabla f(x,y)|$ for $(x,y)$ not precisely on the ridge path. In the interest of clarity, we will refer to these as \emph{spatial oscillations.}

The classic benchmark nonlinear programming problem is typically presented as 
\begin{equation}\label{nlprog1}
\bx^* = \argmin_{\bx \in \mathbb{R}^2} f(x,y).
\end{equation}
We complicate matters by including the amplitude function $T(t)$ such that
\begin{equation}\label{nlprog2}
\bx^* = \argmin_{\bx \in \mathbb{R}^2} E\left[T(t)f(x,y)\right].
\end{equation}
where 
\begin{equation}\label{dynamic_amp_T1}
T(t) = 1 + \frac{3}{4} \cos{(2 \pi t)}.
\end{equation}
Again, for the sake of clarity, we shall refer to oscillations caused by a dynamic amplitudes like \eqref{dynamic_amp_T1} as \emph{temporal oscillations}.

In our first experiment, we attempt to solve our temporally oscillating problem  \eqref{nlprog2} with the standard gradient descent method (Algorithm \ref{Alg2}.)  We initialize at $\bx_0 = (-1.2,1)$, fix $\alpha_i = \delta =  1/500$, $\beta = 0$, enforce a step-size maximum $||\bx_{i+1} - \bx_{i}|| \leq 1 / 4$, and a maximum iteration count of $i_{\text{max}} = 1200$. The gradient is computed by a uniform sampling of $N=15$ antipodal pairs. In Figure \ref{rosenbrockfigs_noise}, we see that the gradient estimates are erroneous far beyond what can be tolerated by the standard algorithm. The figure only depicts steps up to $i_{200}$, since the full path eventually diverges. Increasing the momentum value $\beta$ has no appreciable impact on this outcome.

\begin{figure}
	\begin{center}
		\includegraphics[width=10cm]{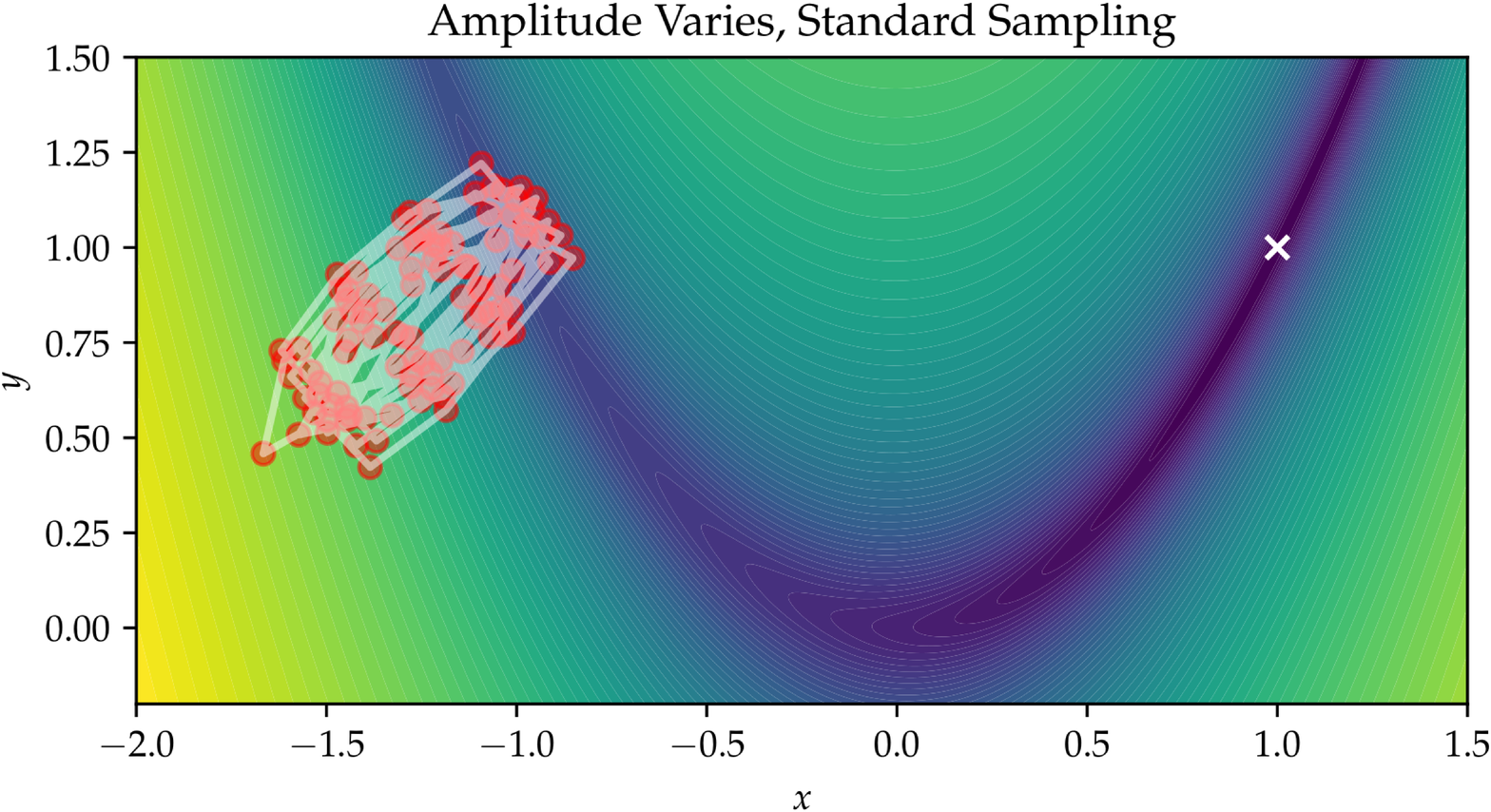} \\
		\hspace{0.5cm} 0 \ \includegraphics[width=6cm]{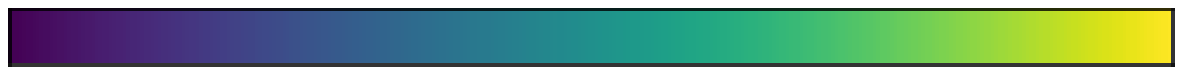} \ 1800
	\end{center}
	\caption{\label{rosenbrockfigs_noise}This figure demonstrates the failure of Algorithm \ref{Alg2} to solve the temporally-fluctuating problem \eqref{nlprog2}. Our plot only considers the first 200 steps, due to an eventual divergence. Each step is depicted by a red dot, connected in sequence by a white line. The white cross in each plot depicts the optimal solution at $(1,1)$. The spatial coordinates and color axis are all non-dimensionalized.}
\end{figure}

In the second experiment, we seek to demonstrate that our amplitude correcting gradient estimate is effective in overcoming the temporal oscillations imposed by \eqref{dynamic_amp_T1}. We accomplish this by comparing the performance of Algorithm \ref{Alg1}, which utilizes the dynamic amplitude correction, on the temporally-oscillating problem \eqref{nlprog2} to the performance of classic gradient descent method in Algorithm \ref{Alg2} on the non temporally-oscillating problem in \eqref{nlprog1}. For each execution we initialize at $\bx_0 = (-1.2,1)$, selecting $\alpha_i = \delta =  1/500$, $\beta = 0$, enforce a step-size maximum $||\bx_{i+1} - \bx_{i}|| \leq 1 / 4$, and a maximum iteration count of $i_{\text{max}} = 1200$. The gradient is computed by a uniform sampling of $N=15$ antipodal pairs. In the temporally oscillating problem, we prescribe a time-step of $h = 1/16$. We provide comparisons with, and without momentum in Figure~\ref{rosenbrockfig}.

In the first row of Figure \ref{rosenbrockfig} we see that without momentum ($\beta = 0$) neither implementation manages to overcome the spatial oscillations. By iteration count $i_{\text{max}} = 1200$, both executions seem to terminate in roughly the same position. In the second row, we see that when momentum is included, both methods overcome the spatial oscillations and reach the global minimum position. When considering the apparent qualitative similarity between these outcomes, in conjunction with the failure demonstrated in Figure \ref{rosenbrockfigs_noise}, we posit that the amplitude corrections are effective in mitigating the temporal oscillations imposed on \eqref{nlprog2}.

\begin{figure} 
	\begin{center}
		\includegraphics[width=10cm]{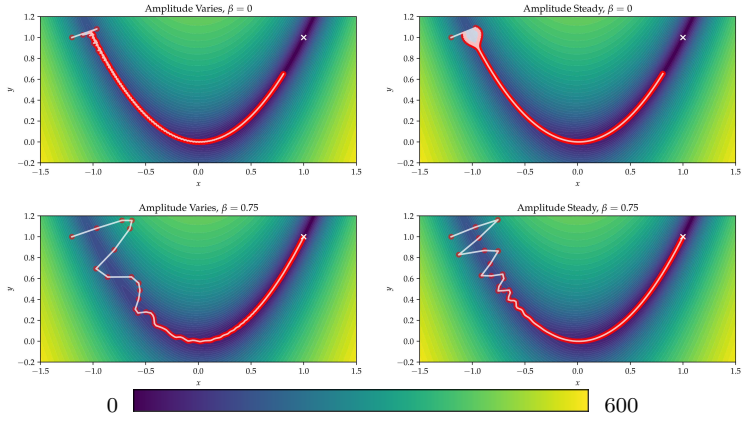}\\
		0 \ \includegraphics[width = 6cm]{colorbar.eps} \ 600
	\end{center}
	\caption{\label{rosenbrockfig} The top row compares the results from Algorithms \ref{Alg1} and \ref{Alg2} to problems \eqref{nlprog2} and \eqref{nlprog1} respectively, with no momentum term ($\beta = 0$). The bottom row makes the same comparison, but selects a momentum term $\beta = 0.75$. Each step is depicted by a red dot, connected in sequence by a white line. The white cross in each plot depicts the optimal solution at $(1,1)$. The spatial coordinates and color axis are again all non-dimensionalized.}
\end{figure}

\subsection{A Convergence Study}
The following theorem provides conditions sufficient for the convergence of the standard differencing gradient in SGD (Algorithm \ref{Alg2}), as well as an error estimate. Proof can be found in \cite{Nguyen2019}.

\begin{thm}[Convergence of SGD with Probability One]\label{sgd_converge} Under the following assumptions,
\begin{enumerate}
    \item[{\bf 1.)}] The objective function $f: \mathbb{R}^n \rightarrow \mathbb{R}$ is $\mu-$strongly convex, i.e., there exists $\mu > 0$ such that
    \begin{equation*}
        f(\bx) - f(\bx') \geq \nabla f(\bx')^T \cdot (\bx - \bx') + \frac{\mu}{2}||\bx - \bx'||^2.
    \end{equation*} 
    \item[{\bf 2.)}] For particular realizations of ${\bf \varsigma}$, the noise-corrupted objective function $\hat{f}(\bx) = f(\bx) + \varsigma$ is L-smooth, i.e., there exists an $L > 0$ such that for any $\bx',\bx \in \mathbb{R}^n$,
    \begin{equation*}
        ||\nabla_\delta \hat{f}(\bx) - \nabla_\delta \hat{f}(\bx')|| \leq L||\bx-\bx'||.
    \end{equation*}
    \item[{\bf 3.)}] The noise-corrupted cost function $\hat{f}$ is convex for every realization of $\varsigma$, i.e., for any $\bx, \bx' \in \mathbb{R}^n$
    \begin{equation*}
        \hat{f}(\bx) - \hat{f}(\bx') \geq \nabla_\delta \hat{f}(\bx')^T \cdot (\bx-\bx').
    \end{equation*}
\end{enumerate}
Then considering Algorithm \ref{Alg2} with step sizes 
\begin{equation*}
    0 < \alpha_i < \frac{1}{2L}, \ \sum_{i=0}^\infty \alpha_i = \infty \ \text{and} \ \sum_{i=0}^{\infty} \alpha_i^2 < \infty,
\end{equation*}
the following holds with probability 1 (almost surely)
\begin{equation*}
    ||\bx - \bx^* ||^2 \rightarrow 0,
\end{equation*}
where $\bx^* = \text{\emph{argmin}}_{\bx \in \mathbb{R}^n} f(\bx)$.
\end{thm}

The following result presents convergence of the stochastic gradient descent method in terms of the error seen in the gradient estimates of the cost function.  
\begin{cor}\label{sgd_conv_rate} Under the same assumptions of Theorem \ref{sgd_converge}, let $\mathcal{E} = \frac{4L}{\mu}$. Initialize Algorithm \ref{Alg2} with step size $\alpha_i = \frac{2}{\mu(t+\mathcal{E})} \leq \alpha_0 = \frac{1}{2L}.$ Then,
\begin{equation*}
    E\left[||\bx - \bx^* ||^2 \right] \leq \frac{16M}{\mu^2} \frac{1}{(t-\tau+\mathcal{E})},
\end{equation*}
for 
\begin{equation*}
    t \geq \tau = \frac{4L}{\mu} \text{\emph{max}}\left\lbrace \frac{L\mu}{M}||\bx_0 - \bx^*||^2,1 \right\rbrace - \frac{4L}{\mu},
\end{equation*}
where $M = 2E\left[||\nabla_\delta \hat{f}(\bx^*) ||^2\right]$ and $\bx^* = \text{\emph{argmin}}_{\bx \in \mathbb{R}^n} \hat{f}(\bx)$.
\end{cor}

We now look to numerically verify the convergence of Algorithm \ref{Alg1} on the problem:
\begin{equation}\label{last_nlp}
    \min_{\bx \in \mathbb{R}^3} E\left[G(\bx,t)\right],
\end{equation}
where the cost function
\begin{equation}\label{cost_f}
    G(\bx,t) = -\left(1 + \frac{3}{4}\cos{\left(2\sqrt{2} \pi t\right)}\right)\left(\bx^T \mathbf{\Sigma} \bx \right) + \varsigma(t),
\end{equation}
with $\mathbf{\Sigma}$ given by
\begin{equation*}
    \mathbf{\Sigma} = \begin{bmatrix}
   2 & -0.5 & 0 \\
   -0.5 & 2 & -0.5 \\
   0 & -0.5 & 2
\end{bmatrix}.
\end{equation*}

We proceed by first demonstrating that an \emph{a priori} accuracy of the SGD algorithm can be written in terms of the asymptotic error of our gradient estimate developed in Theorem \ref{NoisyThm}. This allows us to formalize a parameterization of the error developed in Algorithm \ref{Alg1} as a function of $\delta, h, \sigma,$ and $N,$ and to test the error rates. As before, the additive noise term $\varsigma(t)$ is i.i.d. and $\mathcal{N}(0,\sigma)$. The asymptotic error of the dynamic amplitude-correcting gradient estimates of $G$ are given, in expectation, in Corollary \ref{grad_cor}. Proof of the following comes directly from Theorems \ref{NoiseFreeThm}, \ref{NoisyThm}, and Young's inequality.

\begin{cor} \label{grad_cor} The gradient of the cost function $G(\bx,t)$ in \eqref{cost_f} can be estimated such that given $\sigma, \delta, h> 0,$ where $h^3 < \delta$, there exists positive constants $c_1, c_2,$ and $c_3$ such that
\begin{equation*}
    E\left[||\nabla_{\delta,h}G({\bf 0},t)||^2\right] \leq c_1(1 + N^2)h^2  + c_2 \delta^4 + c_3 \frac{1}{N} \frac{\sigma^2}{\delta^2}.
\end{equation*}
\end{cor}

When we ignore the time-dependent amplitude of the cost function $G$ from \eqref{cost_f}, we note that it was constructed to satisfy the assumptions from Theorem \ref{sgd_converge}. In particular, Assumption 1 is satisfied with $\mu = 2.$ In the calculations to follow, the initial position is $\bx_0 := (1,1,1)$, hence the total distance we intend Algorithm \ref{Alg1} to travel is $|| \bx_0 - {\bf 0}|| = \sqrt{3}$. We also note that since $L$ is sensitive to $\varsigma(t)$, it is not precisely known. For appropriately converging step-sizes $\lbrace \alpha_i \rbrace_{i = 1}^{\infty}$, we will see
\begin{equation*}
    || \bx_i - {\bf 0}||^2\rightarrow 0,
\end{equation*}
with probability 1 and when 
\begin{equation*}
    i > \tau =  \max \left\lbrace \frac{2\sqrt{3}L^2}{E\left[||\nabla_{\delta,h}G({\bf 0},t)||^2\right]}, L \right\rbrace,
\end{equation*}
we see
\begin{eqnarray} \label{alg_Error}
    E \left[|| \bx_i - {\bf 0}||^2 \right] & \leq & E\left[||\nabla_{\delta,h}G({\bf 0},t)||^2\right] \frac{1}{i - \tau} \nonumber 
    \\ & \leq & c_1(1 + N^2)h^2  + c_2 \delta^4 + c_3 \frac{1}{N} \frac{\sigma^2}{\delta^2}.
\end{eqnarray}
\noindent Thus, for steps $i > \tau$, the error seen in \eqref{alg_Error} is proportional to that seen for the gradient estimate in Corollary \ref{grad_cor}. 

These error estimates are verified in a series of Monte Carlo studies. For each parametrization, we repeat and store the results from 30 executions of Algorithm \ref{Alg1}, storing the results in
\begin{equation*}
    \text{err}(\delta,h,\sigma,N) = \left\lbrace||\bx_{i_{\text{max}}, k}||^2 \right\rbrace_{k=1}^{30}.
\end{equation*}
In the first experiment, we fix $\delta$, $\sigma$, and $N$ such that their contributions to the error in \eqref{alg_Error} are several orders of magnitude below our choices for $h$. We further assume that $L \approx ||\mathbf{\Sigma}||_2 = 2 + \sqrt{2}/2,$ which gives for $h$ sufficiently small, that our critical algorithm step $\tau$ is $\mathcal{O}(h^{-2}).$ Selecting a fixed step-size $\alpha_i = \delta$, with a fixed stopping point, trivially satisfies the convergence requirements of Theorem \ref{sgd_converge}. In addition, given our estimate of $\tau$, and the minimum travel distance required, selecting $N=5$, and $\delta =$1/100, we find a choice of $i_{\text{max}} = 500$ to be appropriate. The results of this test confirm the \emph{a priori} rate estimate of $\mathcal{O}(h^2)$, and are presented in terms of the average result over the 30 simulations in Table \ref{h_conv_sim}. 
\begin{table} 
\caption{\label{h_conv_sim} Monte Carlo simulations were used to estimate the accuracy of 500 steps from Algorithm \ref{Alg1}, in terms of $h$. Corollary \ref{grad_cor} suggests we should see error converge at a rate of 2. We fix $N$=5, $\delta$ = 1/100, and $\sigma$ = 1E-5. The other sources of error begin to dominate for choices of $h \leq 1/512$.}
\begin{center}
\begin{tabular}{|ccc|} 
     \hline
      $h$ & $\text{AVG}\left(\text{err}(h)\right)$ &  Rate \\
     \hline
     1/16& 1.7E-4 & - \\
     1/32& 3.2E-5 & 2.41 \\
     1/64& 2.2E-6 & 3.87 \\
     1/128& 5.2E-7 & 2.07 \\
     1/256& 1.1E-7 & 2.28 \\
     1/512& 3.7E-8 & 1.53 \\
     \hline 
\end{tabular}
\end{center}
\end{table}

We proceed similarly in the second experiment. We fix $h = 1/1024$, $N=10$, and $\delta=1/100$, varying $\sigma$. Noting that smallest choice for $\sigma = 1/2560$, we again estimate the critical time-step as $\tau =$ 500. The optimal rates are observed and presented in Table \ref{sig_conv_sim}.

\begin{table} 
\caption{\label{sig_conv_sim} Monte Carlo simulations were used to estimate the accuracy of 500 steps from Algorithm \ref{Alg1}, in terms of $\sigma$. We fix $N$=10, $\delta$ = 1/100, and $h$ = 1/1024.  Corollary \ref{grad_cor} suggests we should see error converge at a rate of 2. The other sources of error begin to dominate for choices of $\sigma \leq 1/2560$.}
\begin{center}
\begin{tabular}{|ccc|} 
     \hline
      $\sigma$ & $\text{AVG}\left(\text{err}(\sigma)\right)$ &  Rate \\
     \hline
     1/80    & 4.0E-4 & - \\
     1/160   & 1.2E-4 & 1.73 \\
     1/320   & 4.0E-5 & 2.04 \\
     1/640   & 9.7E-6 & 2.03 \\
     1/1280  & 2.1E-6 & 2.19 \\
     1/2560  & 5.8E-7 & 1.87 \\
     \hline 
\end{tabular}
\end{center}
\end{table}

For the third, we fix $N$=256, $\sigma$ = 1/2048, and $h$ = 1/2048, varying $\delta.$ We maintain our choice of $i_{\text{max}}$ = 500, presenting the results in Table \ref{delta_conv_sim}. We see rates comparable to the $\mathcal{O}(\delta^4)$ rate. 

\begin{table} 
\caption{\label{delta_conv_sim} Monte Carlo simulations were used to estimate the accuracy of 5000 steps from Algorithm \ref{Alg1}, in terms of $\delta$. We fix $N$=256, $\sigma$ = 1/2048, and $h$ = 1/2048. Corollary \ref{grad_cor} suggests we should see error converge at a rate of 4. The other sources of error begin to dominate for choices of $\delta \leq 1/100.$}
\begin{center}
\begin{tabular}{|ccc|} 
     \hline
      $\delta$ & $\text{AVG}\left(\text{err}(\delta)\right)$ &  Rate \\
     \hline
     0.300    & 5.4E-2 & - \\
     0.210    & 1.98E-2 & 2.90 \\
     0.149    & 5.480E-3 & 3.70 \\
     0.105    & 1.38E-3 & 3.98 \\
     0.074    & 3.32E-4 & 4.11 \\
     \hline 
\end{tabular}
\end{center}
\end{table}

In our final experiment, we test the linear convergence rate of the sampling count parameter $N.$ Fixing $h$ = 1/1024, $\delta$ = 1/100, and $\sigma$ = 0.64, varying $N$. We select $i_{\text{max}} = 500$. In Table \ref{N_conv_sim} we see convergence at a rate slightly better than the expected $\mathcal{O}(N^{-1})$ rate.  

\begin{table} 
\caption{\label{N_conv_sim} Monte Carlo simulations were used to estimate the accuracy of 500 steps from Algorithm \ref{Alg1}, in terms of the inverse sample count $N^{-1}$. We fix $\delta$ = 1/100, $h$ = 1/1024, and $\sigma = $0.64.  Corollary \ref{grad_cor} suggests we should see error converge at a rate of 1.}
\begin{center}
\begin{tabular}{|ccc|} 
     \hline
      $N$ & $\text{AVG}\left(\text{err}(N)\right)$ &  Rate \\
     \hline
     8    & 6.2E-1 & - \\
     16   & 3.6E-1 & 0.77 \\
     32   & 1.7E-1 & 1.05 \\
     64   & 1.0E-2 & 1.86 \\
     128  & 3.3E-3 & 1.67 \\
     256  & 1.1E-3 & 1.58 \\
     \hline 
\end{tabular}
\end{center}

\end{table}

\section{Compound Refractive Lens Alignment on  Simulated XFEL Experimental Beamline}\label{XFEL}
What follows is a proof-of-concept implementation of Algorithm \ref{Alg1} by simulating the alignment of a CRL assembly on a scientific beam-line with a highly dynamic intensity. We begin by developing a model X-ray transmission function from data collected at the Advanced Photon Source (APS) at Argonne National Laboratory \cite{Breckling2021}. This steady-amplitude model is then augmented with a time-dependent intensity function, recorded during an experiment performed at the Pohang Accelerator Laboratory's XFEL facility (PAL-XFEL).

Given that access to XFEL beam-lines is competitive and limited, our goal is to demonstrate the feasibility of our amplitude-correcting SGD approach to overcome the beam intensity fluctuations inherent to XFEL facilities. We break this effort into two parts: the construction of our model cost function, and the results of our implementation of Algorithm \ref{Alg1} using that cost function in settings similar to those seen at PAL-XFEL. 

\subsection{Developing a Model Cost Function}
Let $\Omega_{max} \subset \mathbb{R}^4$ denote the travel limits for the four stepper motors that determine the orientation of the CRL. For a given orientation $\bx = (x,y,r_x,r_y) \in \Omega$, let the resulting image deposited on the detector panel be denoted as $I(\bx),$ or simply $I$ when convenient; see Figure \ref{hugh_image}. Further, we describe position of a given pixel by its indices $I_{i,j}.$ Example detector images are shown in Figures \ref{sensor_imgs}(a) and (b).

Let $\xi(I)$, $\mu(I)$ and $\sigma(I)$ denote the median, mean, and standard deviation, respectively, of the pixel values of the image $I$. We then constrain
$I$ to a selected region of interest (ROI) % of $I$, 
defined as
$$\hat{I}_M := \left\lbrace I_{i,j} \in I \ \Big| \ |I_{i,j}-\xi(I)| > M \times \sigma(I) \right\rbrace,$$
where $M > 0$ is a user-selected threshold parameter. In practice, we found $M=2$ to be a good choice. Figure \ref{sensor_imgs} (c) and (d) highlight the corresponding ROIs, $\hat{I}_M.$ 

\begin{figure} 
    \begin{center}
    \includegraphics[width=2cm]{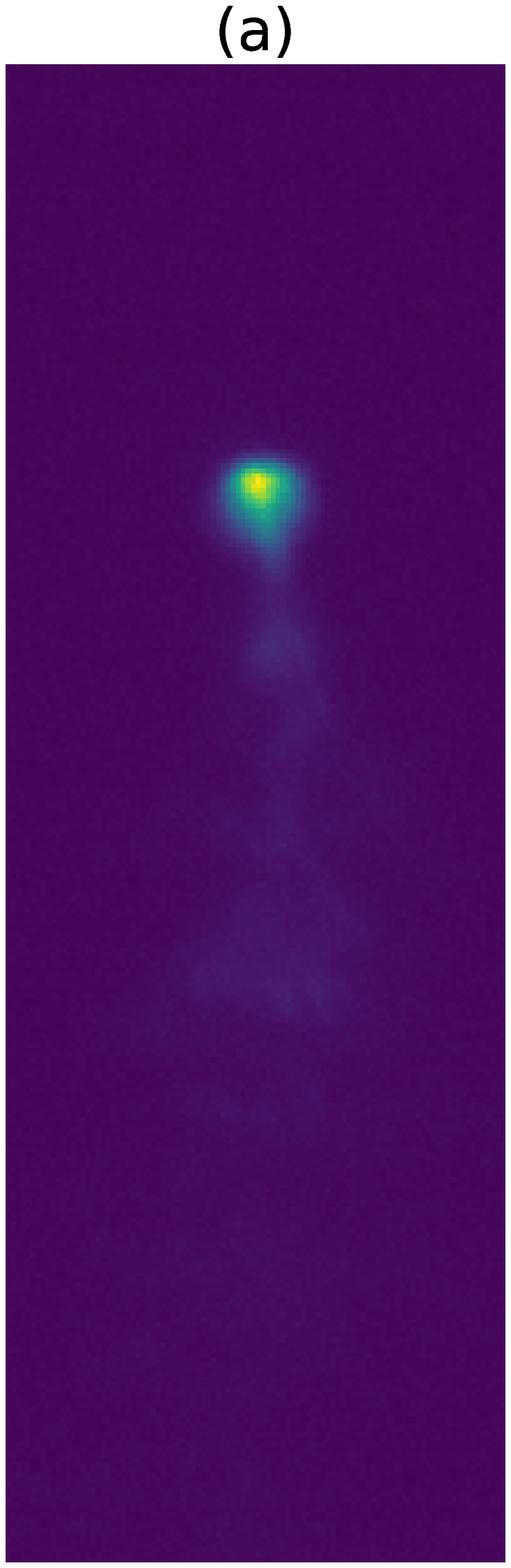} \ 
    \includegraphics[width=2cm]{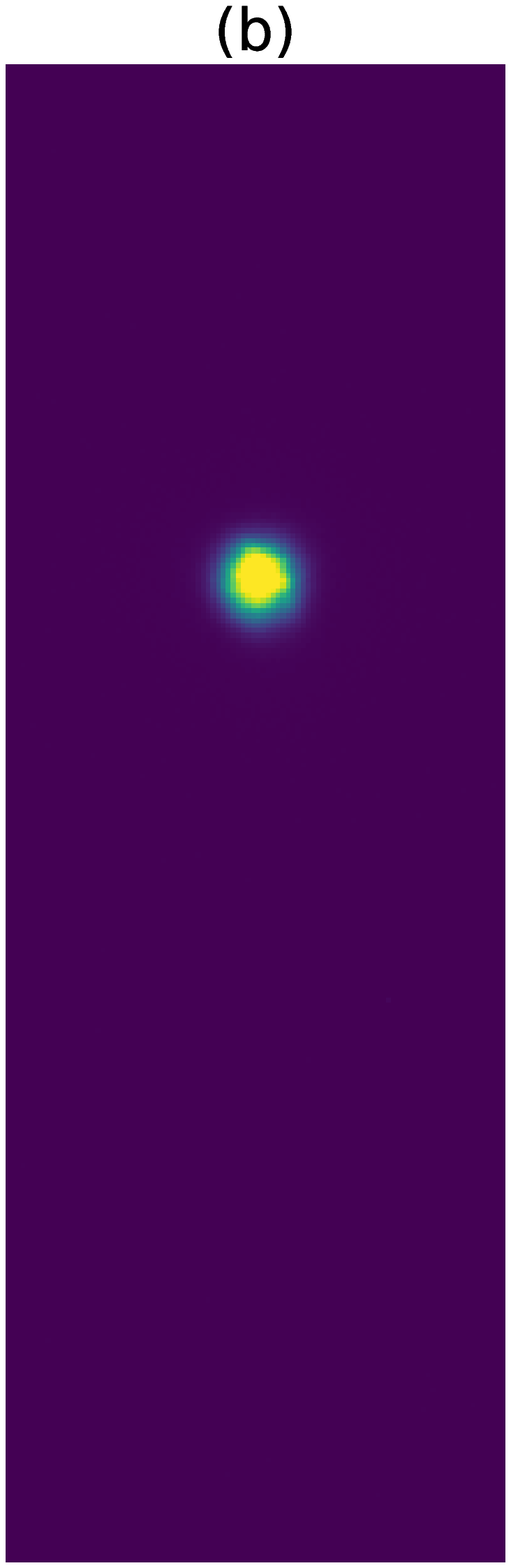} \ 
    \includegraphics[height=5.75cm]{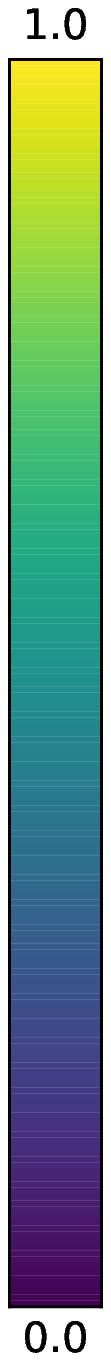} \hspace{1cm} 
    \includegraphics[width=2cm]{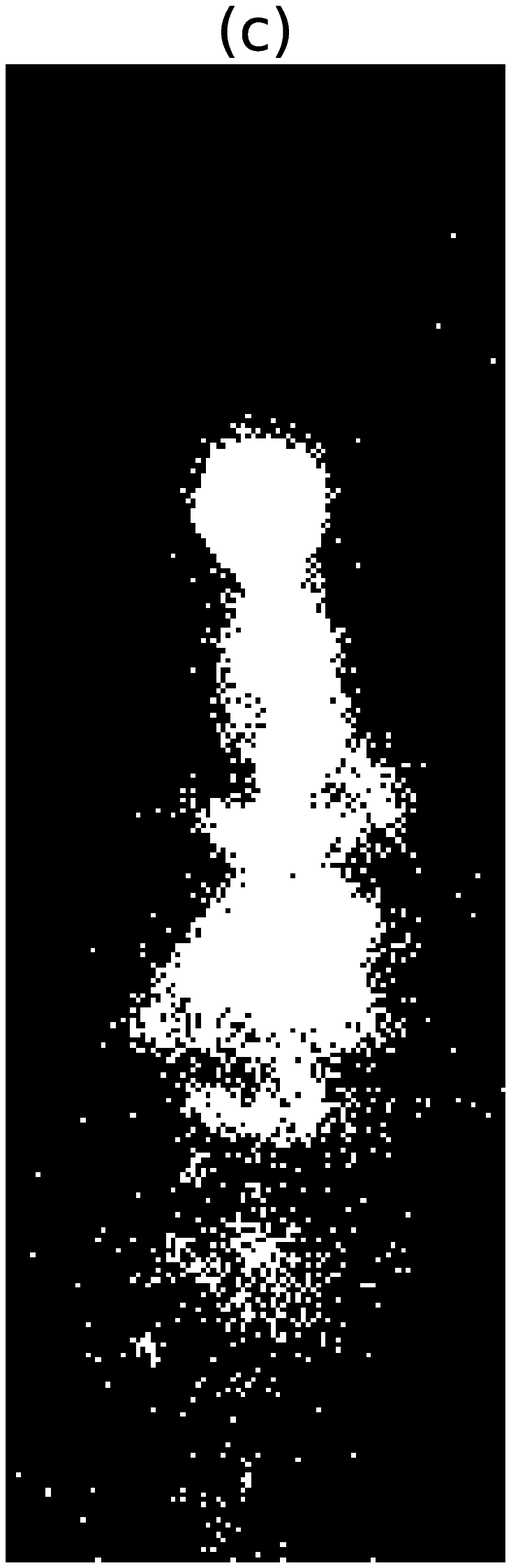} \ 
    \includegraphics[width=2cm]{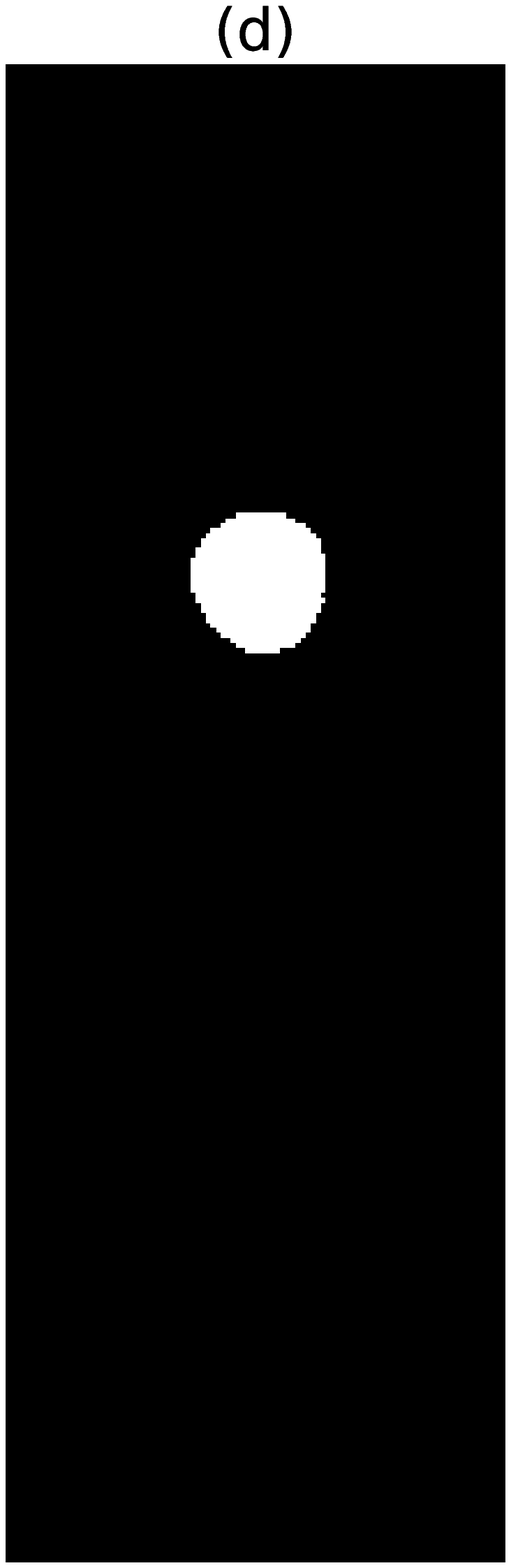} \ 
    \includegraphics[height=5.75cm]{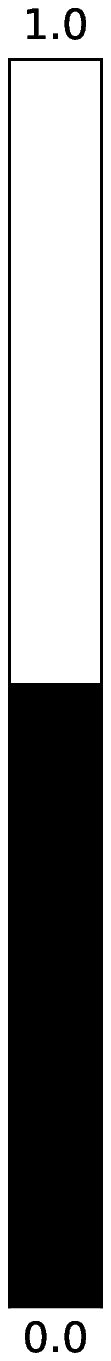}
    \caption{Figure (a) is a cropped region collected from the imaging sensor when the CRL was poorly aligned. Figure (b) is the same cropped region, but shows the result from a well-aligned CRL. The images are shown on the same color axis, after feature normalizing against the maximum pixel value recorded. Figures (c) and (d) are binary images depicting the pixels identified in the ROI for Figures (a) and (b) respectively. \label{sensor_imgs}}
    \end{center}
\end{figure}

In the synchrotron experiments performed at the Advanced Photon Source in \cite{Breckling2021}, a set of coordinates found by manual alignment were defined as the ground-truth to provide the ``well-aligned'' position of the CRL. We denote that position as $\bx^* = (x^*, y^*, r_x^*, r_y^*)$. This ground truth served two purposes. First, we were then able to define a feature scaling such that our metric of X-ray transmission, in terms of CRL orientation, 
\begin{equation}\label{crude_cost}
    f(\bx; M) := \mu\left(\hat{I}_{M}(\bx)\right),
\end{equation}
had a maximal value of 1. Second, the ground-truth position allowed us to establish a four-dimensional rectangular region $\hat{\Omega} \subset \Omega_{max}$ around the best point that contained the support of $f$ above the noise floor. With this ground-truth and 4D window, we then collected several raster scans of $f(\hat{\Omega}; M=2)$. We make use of a full four-dimensional scan, and two high-resolution, independent, 2-dimensional raster scans of $\hat{\Omega}$.

\begin{figure}
    \begin{center}
    \includegraphics[width=4cm]{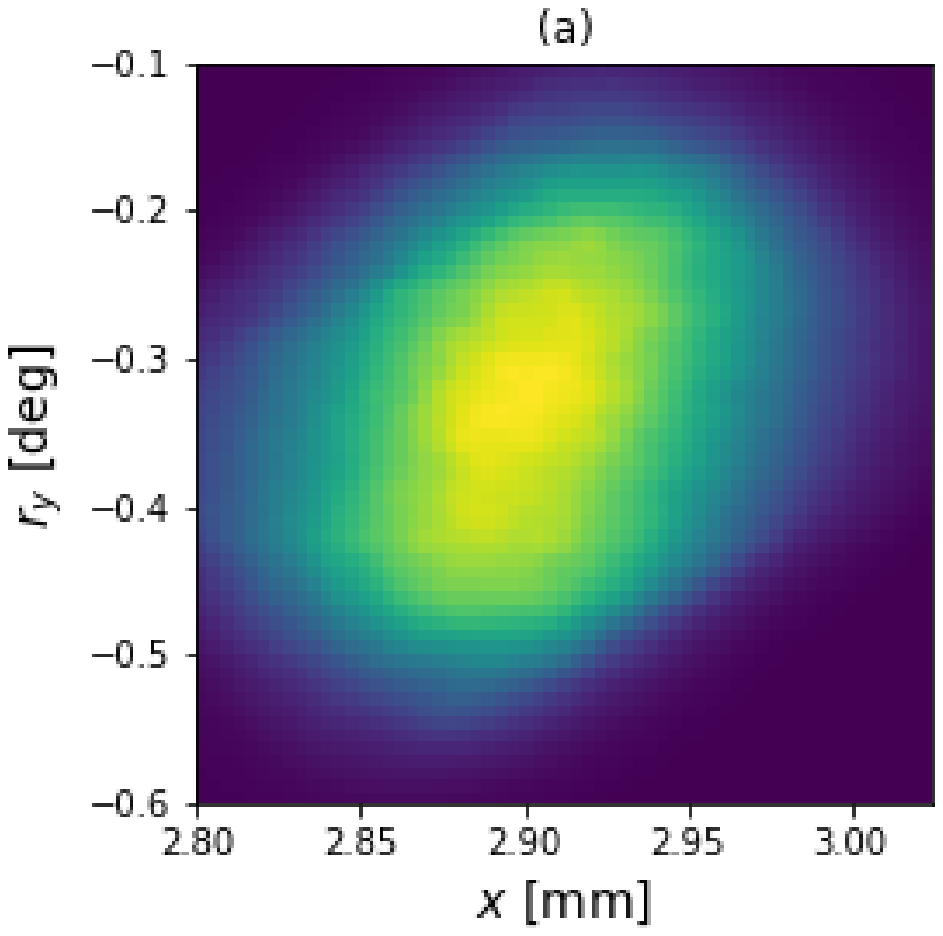} \ 
    \includegraphics[width=4cm]{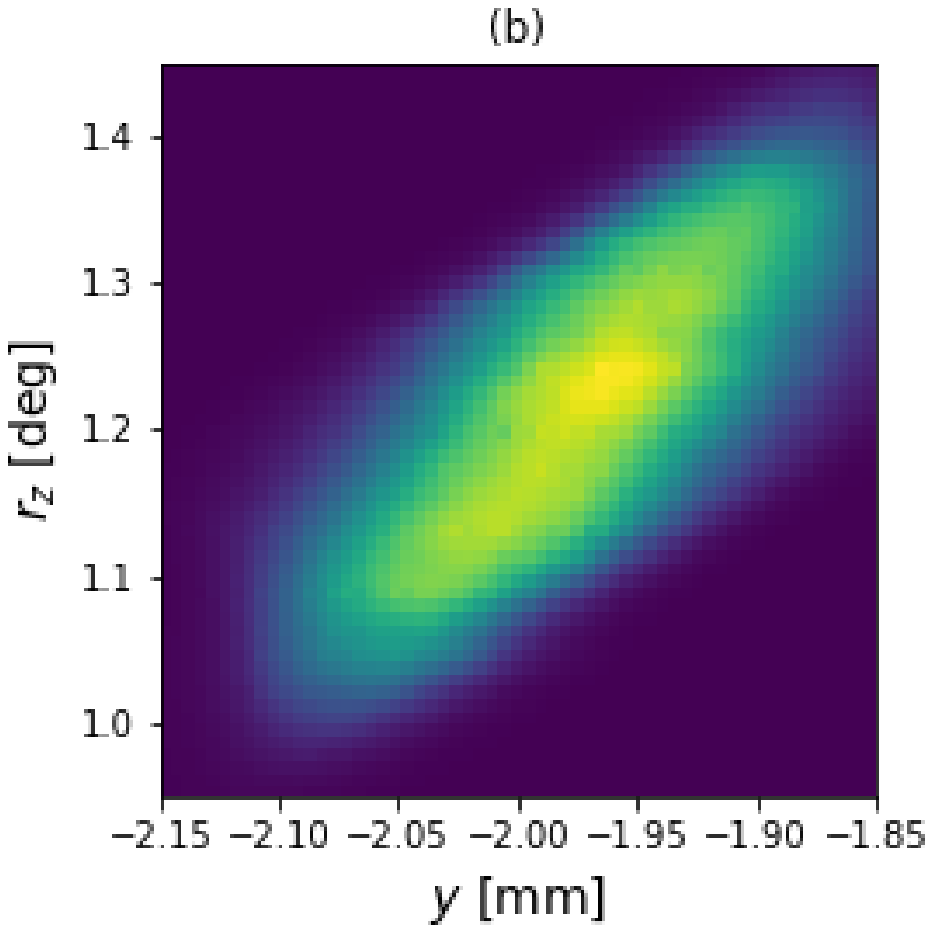} \\
    0 \ \includegraphics[width = 6cm]{colorbar.eps} \ 1
    \caption{Two 2D raster scans of $f$ in $\hat{\Omega}$ are depicted above. Both figures are mutually min-max normalized, and plotted on the same color axis.\label{intensity_plots}}
    \end{center}
\end{figure}

Assuming a steady beam amplitude, it follows from the model developed by Simons \textit{et al.} that an idealized transmission function $f:\mathbb{R}^4 \rightarrow \mathbb{R}^+$ is given by a 4-variate Gaussian distribution \cite{simons17}. We generalize that model as 
\begin{equation} \label{regression}
    f_{\text{Simons}}(\bx; a,b,\mathbf{A}, \hat{\bx}) = a\exp{\left(-(\bx-\hat{\bx})^T  \mathbf{A}  (\bx-\hat{\bx}) \right)} + b,
\end{equation}
where $a,b \in \mathbb{R}$, the matrix $\mathbf{A} \in \mathbb{R}^{4\times4}$ is symmetric, and $\hat{\bx} \in \mathbb{R}^4$ is the position associated with optimal lateral alignment. Fitting the four-dimensional raster scan data $f(\hat{\Omega}; M=2)$ to Simons' model \eqref{regression} gives the idealized X-ray transmission $f_{\text{Simons}}^*(\bx)$. We present two, two-dimensional slice views of $f_{\text{Simons}}^*(\bx)$ in Figure \ref{f_simons_plots}. 

\begin{figure}
    \begin{center}
    \includegraphics[width=4cm]{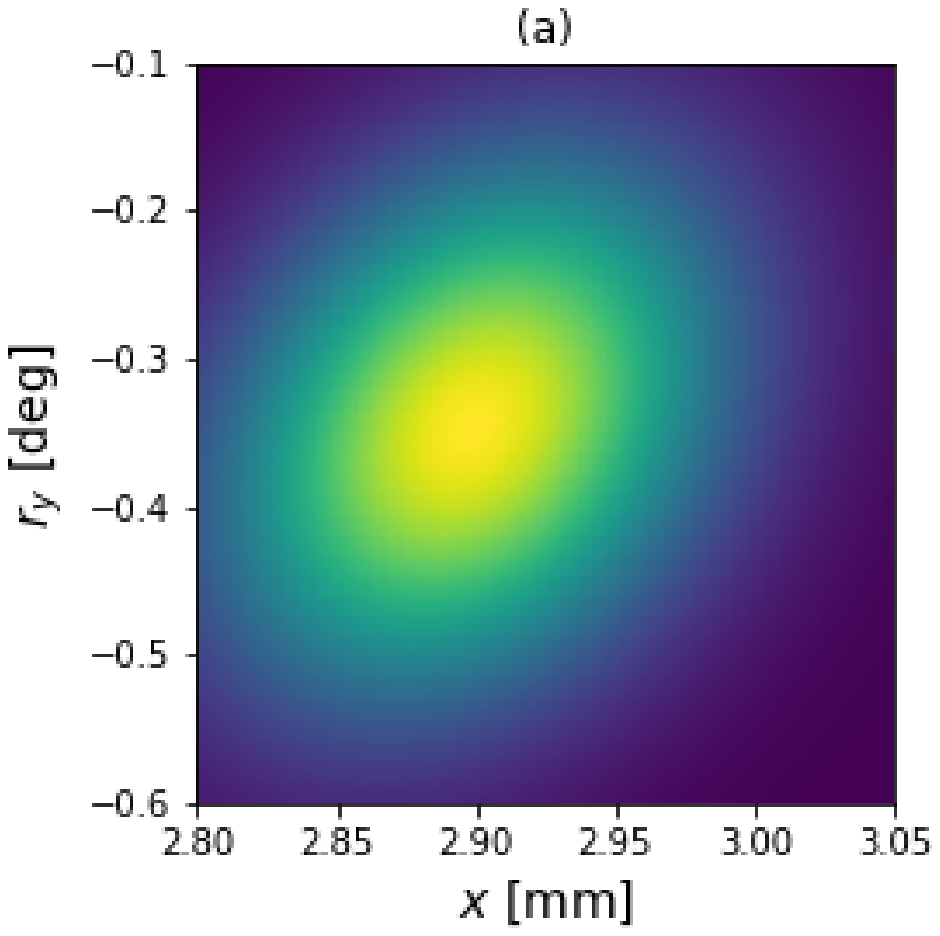} \ 
    \includegraphics[width=4cm]{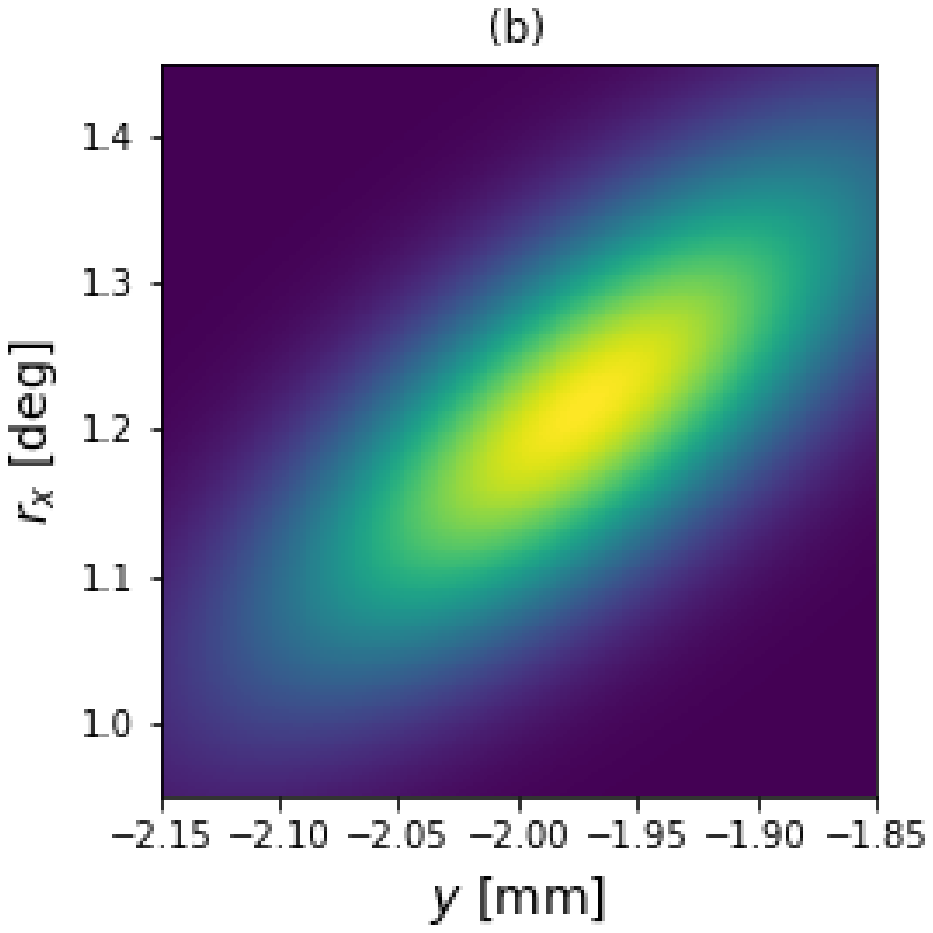} \\
    0 \ \includegraphics[width = 6cm]{colorbar.eps} \ 1
    \caption{Depicted here are 2D slices selected from $f_{\text{Simons}}^*(\hat{\Omega})$. The aspect ratios were selected to agree with Figures (a) and (b) from Figure \ref{intensity_plots}. \label{f_simons_plots}}
    \end{center}
\end{figure}

To model the noise functions that are characteristic to the XFEL light sources, we include additive measurement noise as $\varsigma_\Omega(t)$. Let $\text{diam}(\hat{\Omega})$ denote the maximal \text{diam}eter of the set $\hat{\Omega}$. We collected a sampling $\mathcal{S} = \lbrace \bx_i \rbrace_{i=1}^{500}$ such that for every orientation $\bx_i$, $||\bx_i - \bx^*|| > \text{diam}(\hat{\Omega})$.  We found that $\sigma(f(\mathcal{S};M=2)) \approx 4.5 \times 10^{-3}.$ We then model the time-series of additive noise $\varsigma_\Omega(t)$ as i.i.d. and $\mathcal{N}(0, 4.5 \times 10^{-3})$. 

We additionally consider fluctuations that occur because of pointing jitter (from the SASE generation scheme) \cite{kang2017hard}. We assume the position and direction of the beam may randomly fluctuate as a function of the beam's divergence profile, which was estimated at the APS to be 6.5$\times 10^{-3}$ Radians. We account for jitter in our model as random perturbations of the orientation vector $\bx$ in the $r_x$ and $r_y$ directions. Further, we expect that the beam will jitter randomly within 10\% of the beam-divergence. In doing so, we define 
\begin{equation*}
    \Theta(t) = (0, 0, \theta_x(t), \theta_y(t))
\end{equation*}
where $\theta_x$ and $\theta_y$ are respectively i.i.d and $\mathcal{N}(0,6.5\times 10^{-4})$.

We lastly introduce the fluctuating intensity of the beam over time. To this end, we utilize the measured shot-to-shot intensity values recorded at the PAL-XFEL facility, which was recorded using a quadrant beam position monitor (QBPM) at 30 Hz \cite{DresselhausMarais2020}. We feature-scale the raw pulse-to-pulse time-series data by normalizing the full signal against the mean recorded value. This scaled signal is written as $T_{\text{PAL}}(t,\kappa)$ where $\kappa$ determines the number of pulses averaged during a data collection event. In Figure \ref{intensity_over_time} we show $T_{\text{PAL}}(t,1)$ in dark gray, $T_{\text{PAL}}(t,8)$ in light gray, and $T_{\text{PAL}}(t,264)$ in red. The mollified signals at $\kappa=8$ and $\kappa=264$ respectively represent the average beam intensity over a sampling interval, and the amount of time required to collect all samples necessary to compute the amplitude-correcting gradient. 

\begin{figure}
    \centering
    \includegraphics[width=12cm]{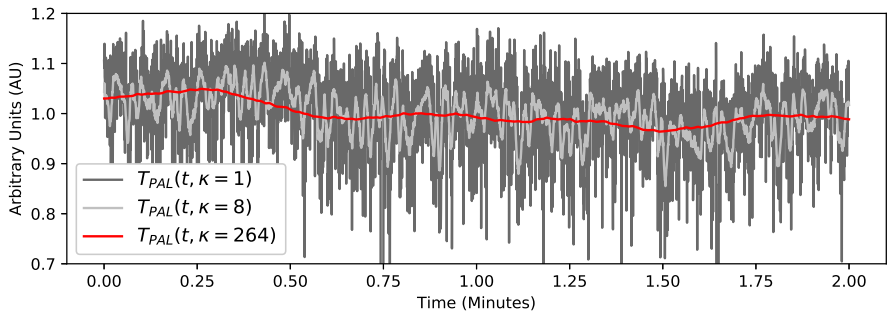}
    \caption{This figure depicts a two-minute interval of the signals $T_{\text{PAL}}(t,\kappa=1)$ in dark gray, $T_{\text{PAL}}(t,\kappa=8)$ in light gray, and $T_{\text{PAL}}(t,\kappa=264)$ in red. All signals are feature normalized by their mean value. The data was recorded at the PAL-XFEL facility. {\color{green} \cite{DresselhausMarais2020}} \label{intensity_over_time} }
\end{figure}

Our full cost function model is hence written and evaluated as
\begin{equation}\label{xfel_cost}
    G_{\text{XFEL}}(\bx,t; \kappa) = -T_{\text{PAL}}(t,\kappa) f_{\text{Simons}}^*(\bx + \Theta(t)) + \varsigma_\Omega(t).
\end{equation}

\subsection{Solving the CRL Alignment Problem}
We now endeavor to study the performance of Algorithm \ref{Alg1} on our model of the CRL alignment problem
\begin{equation*}
    \min_{\bx \in \hat{\Omega}} E\left[G_{\text{XFEL}}(\bx,t;\kappa = 8)\right], \ \forall t > 0.
\end{equation*}
Our goal is to identify a range of nominal parameter choices for Algorithm \ref{Alg1} that can be implemented as a starting point at an XFEL facility. 

We begin by noting that when sampling \eqref{xfel_cost} to estimate $\nabla f_{\text{Simons}}^*(\bx)$ as per Definition \ref{def_grad}, we consider $N=8$ quasi-uniformly distributed antipodal pairs in our differencing stencil. We select our effective integration time interval for the camera to be $h_{\text{cam}} = 8/30$ seconds, and establish the full time interval required to complete the scheme as $\mathcal{T}_{h,N} := [t_0, (4N+1)h + t_0].$ Further, we make use of the estimate
\begin{equation*}
    \mu_T = T_{\text{PAL}}(t_0 + 264/30,264),
\end{equation*}
where $t_0$ is the moment we began estimating the gradient.

For each execution of Algorithm \ref{Alg1} that follows, the stopping condition is established to be a maximal iteration count $i_{\text{max}}.$ No other stopping conditions are considered. Additionally, we conceptualize our initial gradient sphere radius $\alpha_0$ as some multiple $C r,$ where $r = ||\bx_0 - \hat{\bx}||,$ though we don't expect users to know what $r$ is \textit{a priori.} At each step $i$, the gradient sampling radius $\alpha_0$ is scaled by a cooling factor such that
\begin{equation*}
    \alpha_i = \frac{\alpha_0}{(1 + i)^\gamma},
\end{equation*}
where $\gamma > 0$ and fixed. Further, we enforce a maximum step size $||\bx_i - \bx_{i-1}|| \leq \delta_i = \alpha_i.$ Given that the true distance $r$ is unknown upon initialization, the executions that follow are intended to identify a performance relationship between $\alpha_0$ with respect to $r$, $\gamma$, and the stopping condition.

We demonstrate a single execution of Algorithm \ref{Alg1} with an initial position $\bx_0$ selected randomly a distance of $r = 0.4$ from $\hat{\bx}.$ We note that this Euclidean distance is significantly further away from $\hat{\bx}$ than the positions selected during the manually-tuned rough alignments completed during data collections at the more stable synchrotron source at APS \cite{Breckling2021}. We fix $\gamma = 0.3$, select $\alpha_0$ according to the distance scalar $C = 3.0$, assign a momentum term $\beta = 0.15,$ and set the stopping condition to $i_{\text{max}} = 100$ iterations. In addition to the time required to collect the image data from the camera sensor $h_{\text{cam}}$, we need to include an estimate of the time required to move the four stepper motors, and process the data. We assume $h_{\text{move}}=5/30s$. Given the full time interval time interval $h_{\text{total}} = h_{\text{cam}} + h_{\text{move}} = 13/30s$, the total execution time assumed necessary to reach 100 iterations is 
\begin{equation*}
    (4N+1) \times i_{\text{max}} \times h_{\text{total}} = 1430s,
\end{equation*}
 or 23.8 minutes. A figure depicting the particular route taken is presented in the 2D projections shown in Figure \ref{single_exec}.

\begin{figure}[h]
\centering
\includegraphics[width=9cm]{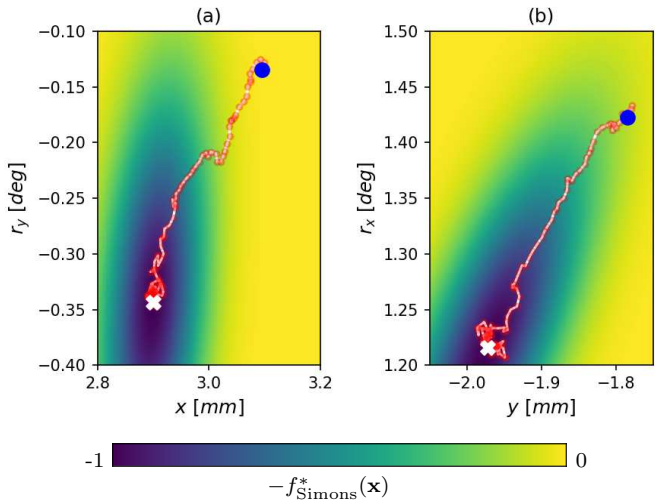} \\
\ \ -1 \includegraphics[width=6cm]{colorbar.eps} 0 \\
$-f_{\text{Simons}}^*(\bx)$
\caption{\label{single_exec}%The two figures below depict a
A single execution of Algorithm \ref{Alg1} with an initial position $\bx_0$ selected randomly at a distance $r = ||\bx_0 - \hat{\bx}|| = 0.4.$ The initial step-size $\alpha_0$ is fixed to $3r = 1.20,$ and scaled with each step by the cooling parameter $\gamma = 0.3.$ The slices depicted in (a) and (b) use the optimal off-axis values in $\hat{\bx}$. The blue dots depict the initial position projected onto the respective 2D planes, the white crosses depict the optimum alignment coordinates $\hat{\bx}$, while the red dots depict the 100 positions $\bx_i$. Each position is connected sequentially by a white line.}
\end{figure}

Next, we present the result of three Monte Carlo experiments. We maintain the parameter choices established in the execution above, varying only the stopping condition $i_{\text{max}} = 50, 100,$ and $200$. Each Monte Carlo executes Algorithm \ref{Alg1} to completion 100 times, varying the initial position randomly on  $\partial \mathcal{B}_r(\bx_0)$ where $r = 0.4.$ The results depicted in Figure \ref{final_MC_1} demonstrate the expected convergence behavior for those well-selected parameters. 

\begin{figure}
    \centering
    $i_{\text{max}} = 50; (11.9 \ \text{Minutes})$\\
    \includegraphics[width=7cm]{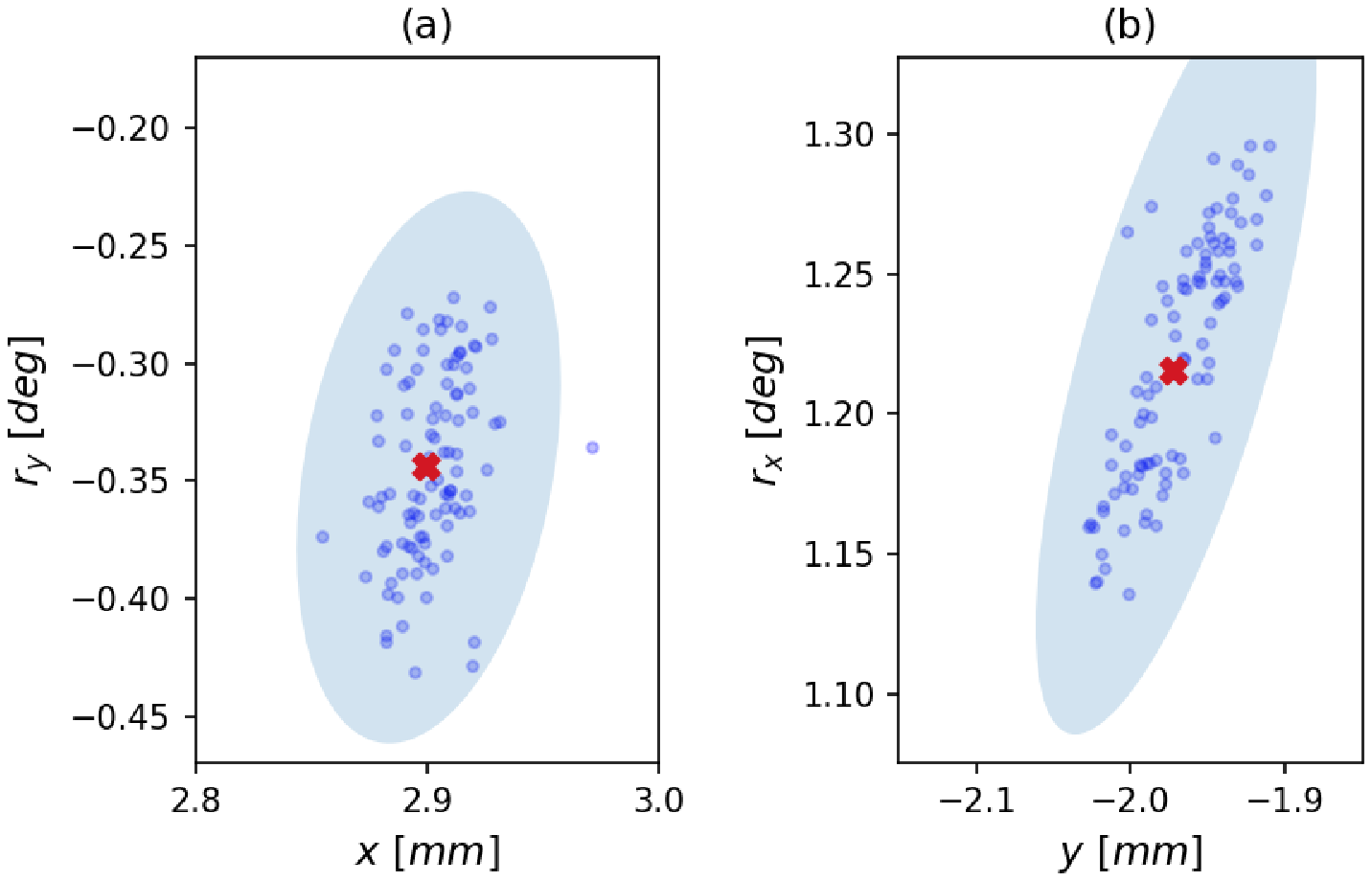}\\
    $i_{\text{max}} = 100; (23.8 \ \text{Minutes})$ \\
    \includegraphics[width=7cm]{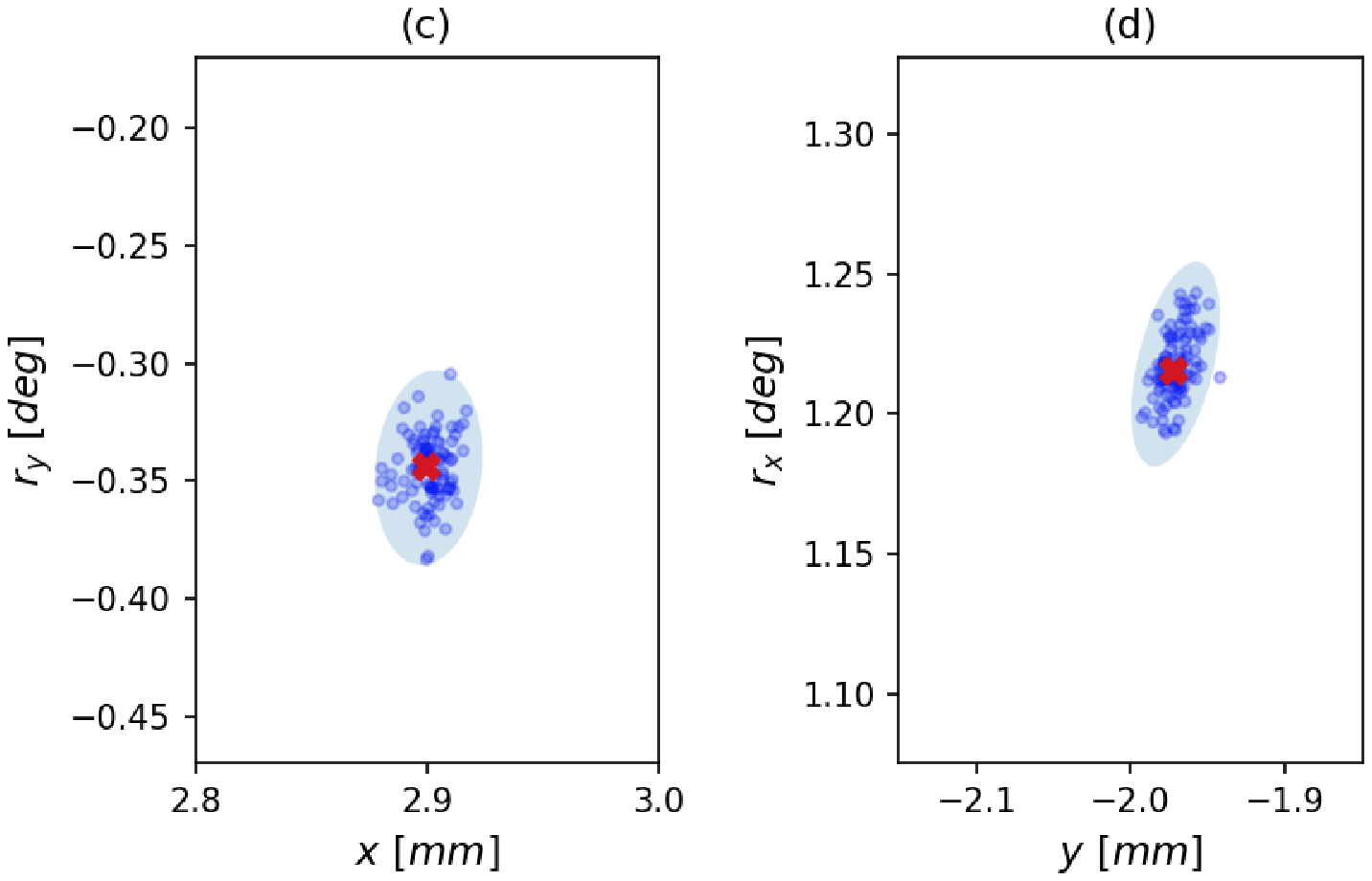}\\
    $i_{\text{max}} = 200; (47.7 \ \text{Minutes})$ \\
    \includegraphics[width=7cm]{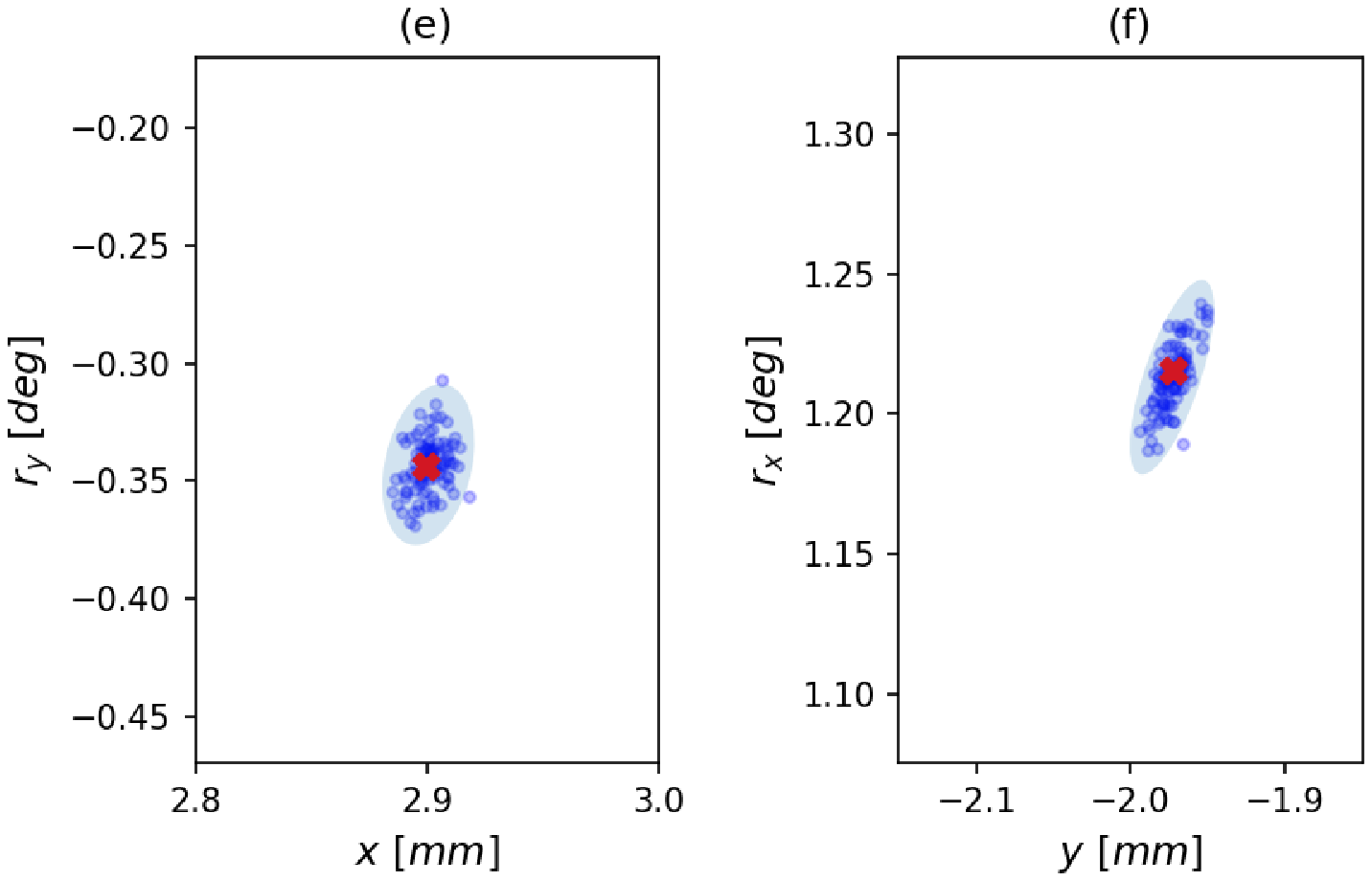}\\
    \caption{\label{final_MC_1}Depicted above are the results of three Monte Carlo simulations, wherein Algorithm \ref{Alg1} is executed 100 times to solve the synthetic CRL alignment problem, varying the stopping condition $i_{\text{max}}$. The point $\hat{\bx}$ is depicted in each figure as a red cross. Figures (a) and (b) show the spatial distribution of results when $i_{\text{max}} = 50$. Similarly, Figures (c) and (d) denote the results when $i_{\text{max}} = 100$, and Figures (e) and (f) show $i_{\text{max}} = 200$. The result of a particular execution is shown as a dark blue dot. The blue ellipses highlight the 99.3\% uncertainty region.}
\end{figure}

What remains to be assessed is performance as a function of the user's choice of step size, and cooling parameter. We consider two values for the cooling parameter $\gamma$,  five initial step-size scales $C$, and six stopping conditions $i_{\text{max}}.$ For each particular set of parameters, Algorithm \ref{Alg1} is executed 100 times, where the starting position $\bx_0$ is again sampled randomly from $\partial \mathcal{B}_r(\hat{\bx})$. These regions demonstrate a collection of parameters that tend to reliably converge under an hour ($i_{\text{max}} < 200$ iterations.)  

\begin{figure}
    \centering
    \includegraphics[width=5cm]{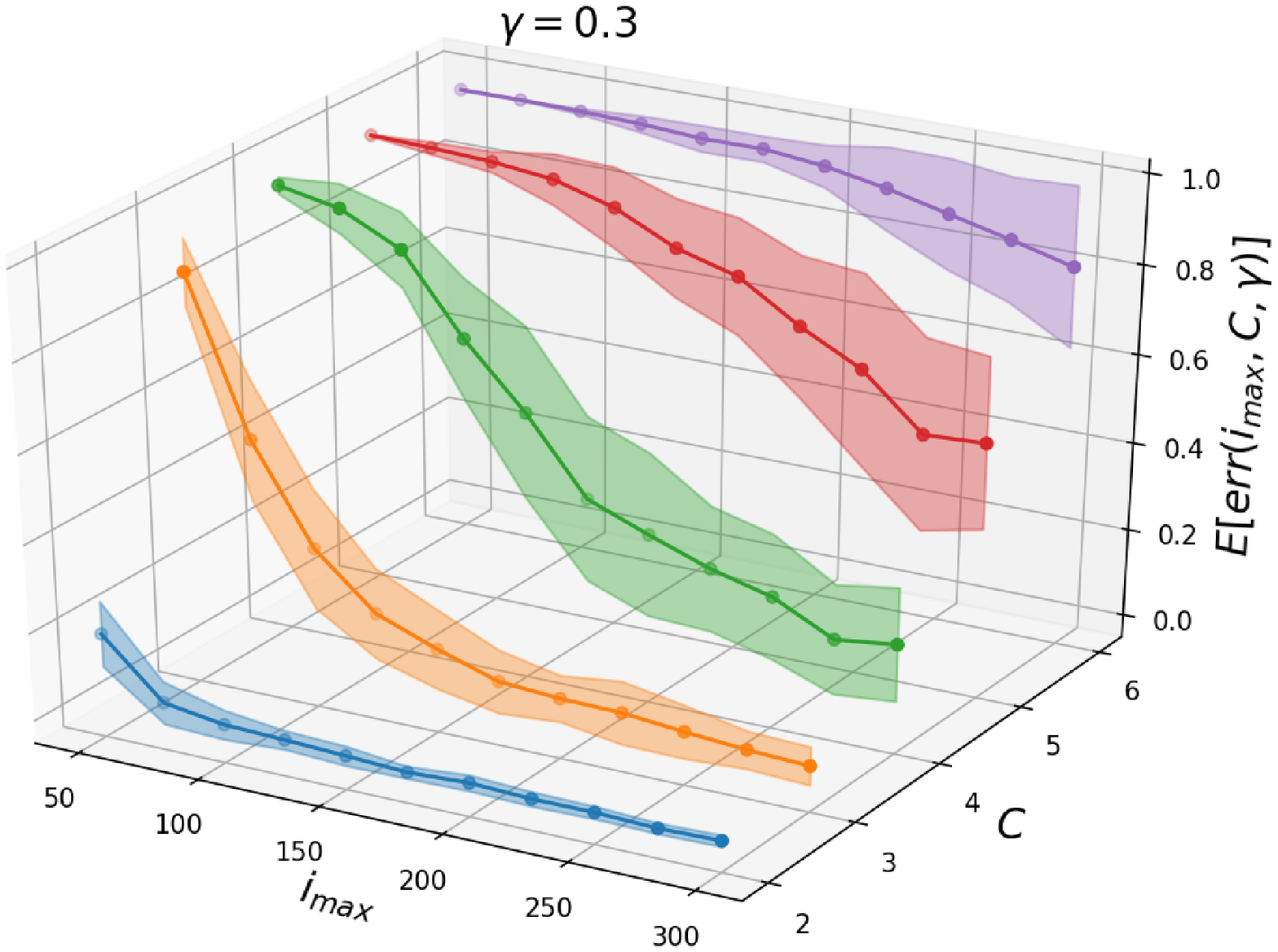} \ 
    \includegraphics[width=5cm]{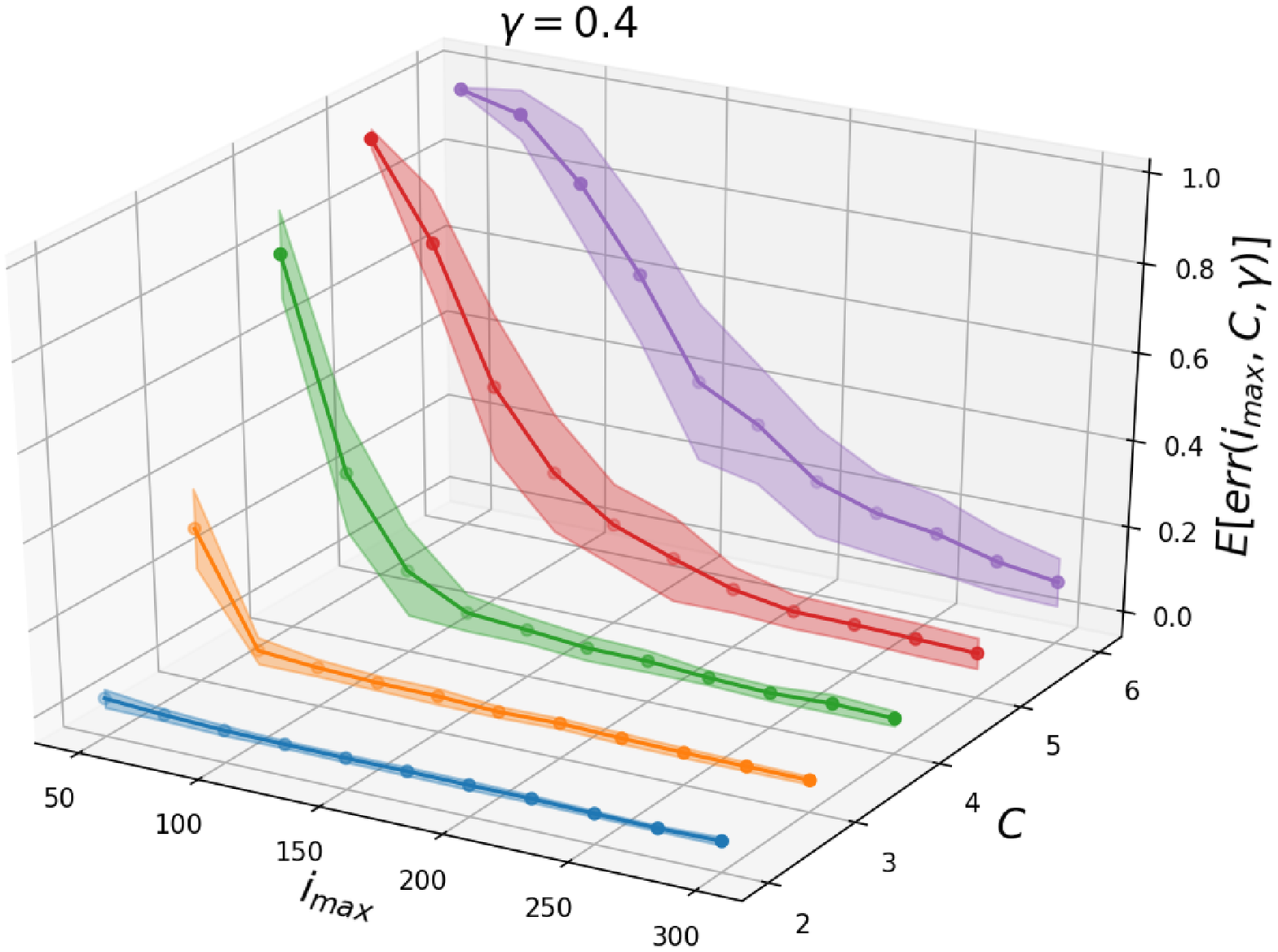}\\
    \caption{Depicted above are the results from 60 Monte Carlo studies varying $i_{\text{max}}$, the cooling parameter $\gamma$, and gradient sphere-radius interval $\alpha_0$. We vary $\alpha_0$ as a multiple of the initial position's distance from ground truth $C||\bx_0 - \hat{\bx}|| = Cr.$ Figures (a) and (b) fix $\gamma$ as 0.3 and 0.4, respectively. Each Monte Carlo simulation executes Algorithm \ref{Alg1} a total of 100 times; the average value of which is depicted normalized by $r$, and depicted as the vertical height. The shaded regions above and below the interpolated lines represent the standard deviation trend, in terms of the 4D Euclidean distance. \label{final_all}}
\end{figure}

While additional parameters remain to be thoroughly studied, namely the momentum term $\beta,$ we found that choices of $\beta > 0.15$ tended to perform poorly over longer periods of computation time. In particular, when $i_{\text{max}} > 50$ we saw no apparent improvement to performance. Given that the settings identified above demonstrate convergence that tends to improve with additional computation, an attractive behavior for an unsupervised optimization method, we advise being conservative with $\beta.$ We additionally note that our choice to equate the maximal step-size with the gradient \text{diam}eter was born out of observation. Choices of $\alpha_i$ substantially larger than $\delta_i$ frequently resulted in failure over longer time intervals. Lastly, we observed that selecting $\gamma$ too large tended to collapse $\alpha_i$ too quickly, which was also detrimental to long-time performance. Conversely, selecting $\gamma$ too small tends to result in slow convergence. 

\section{Conclusions}\label{conclusions}
The motivation for this work was encountered while attempting to automate the task of laterally aligning optics at an X-ray Free-Election Laser (XFEL) facility. These facilities, in aggregate, are capable of generating extremely bright pencil beams of X-ray light, but from moment to moment that brightness fluctuates in time. If not for the stochastic noise sources and the apparent intensity fluctuations, the task of orienting beam-line optics reduces to a rather simple minimization problem \cite{simons17}. While the apparent level of stochastic measurement noise is certainly tractable for many stochastic descent methods, the intensity fluctuations are so severe that they required a separate, independent treatment. 

In this paper, we introduced a differencing scheme to estimate the gradient of a cost functional potentially corrupted by both stochastic noise and independent amplitude fluctuations. We assume that only one position in the search space can be measured at any particular moment in time. Thus, any finite differencing scheme is going to require procedurally moving from point to point, recording each intensity along with the corresponding position and time. In this scheme, we account for the fluctuating amplitude by introducing additional samples at a single, fixed location central to the differencing stencil. By alternating these samples sequentially in time, separating the resulting data post-hoc provides a proportional estimate of the functional's amplitude. This additional signal is then interpolated along the time axis, and subtracted from the corresponding signal generated sequentially by the finite differencing stencil. When well-sampled in space and time, this method is effective at detrending those measurements. 

We included a detailed error analysis of this amplitude-correcting gradient estimate, as well as numerical benchmarking of its performance in nonlinear programming problems solved with SGD. Additionally, given that access to XFEL facilities are highly limited, we included a proof-of-concept implementation of an amplitude-correcting SGD method that we believe shows promise. In doing so, we identified regions of parameter choices that will likely be effective in a similarly-configured apparatus. 

\section*{Acknowledgments} 
This manuscript has been authored in part by Mission Support and Test Services, LLC, under Contract No. DE-NA0003624 with the U.S. Department of Energy, National Nuclear Security Administration (DOE-NNSA), NA-10 Office of Defense Programs, and supported by the Site-Directed Research and Development Program. The United States Government retains and the publisher, by accepting the article for publication, acknowledges that the United States Government retains a non-exclusive, paid-up, irrevocable, world-wide license to publish or reproduce the published content of this manuscript, or allow others to do so, for United States Government purposes. The U.S. Department of Energy will provide public access to these results of federally sponsored research in accordance with the DOE Public Access Plan. The views expressed in the article do not necessarily represent the views of the U.S. Department of Energy or the United States Government. DOE/NV/03624--1406.

Portions of this work were performed at High Pressure Collaborative Access Team (HPCAT; Sector 16), Advanced Photon Source (APS), Argonne National Laboratory. HPCAT operations are supported by the DOE-NNSA's Office of Experimental Sciences.  The Advanced Photon Source is a DOE Office of Science User Facility operated for the DOE Office of Science by Argonne National Laboratory under Contract No. DE-AC02-06CH11357

Sunam Kim, Sangsoo Kim, and Daewoong Nam would like to acknowledge support from the National Research Foundation of Korea (NRF), specifically NRF-2019R1A6B2A02098631 and NRF-2021R1F1A1051444.

Part of this work was performed under the auspices of the U.S. Department of Energy by Lawrence Livermore National Laboratory under Contract DE-AC52-07NA27344. We also acknowledge the support of the Lawrence Fellowship in this work.

% BibTeX users please use one of
%\bibliographystyle{spbasic}      % basic style, author-year citations
%\bibliographystyle{spmpsci}      % mathematics and physical sciences
\bibliographystyle{spphys}       % APS-like style for physics
%\bibliography{}   % name your BibTeX data base

% \bibliographystyle{elsarticle-num} 
 \bibliography{ref}

%% else use the following coding to input the bibitems directly in the
%% TeX file.

% \begin{thebibliography}{00}

% %% \bibitem{label}
% %% Text of bibliographic item

% \bibitem{}

% \end{thebibliography}

\appendix
\section{Accuracy Estimates}
\label{AccuracyEst}
Below we have appended error estimates of the amplitude-correcting gradient tool presented in this manuscript. Accuracy is discussed in terms of both spatial and temporal discretizations, as well as the case where Gaussian, i.i.d. noise is present.  

\begin{lemma}[Error Estimate of a Single Directional-Derivative Stencil on Smooth Functions]\label{mainLemma} Under the same assumptions of Theorem \ref{NoiseFreeThm}, we have there exists a $\mathcal{C}^*>0$ such that
	\begin{equation*}
	\left| \left(\be_k \cdot \nabla \right) f(\bx_c) - \left(\be_k \cdot \nabla_{\delta,h} \right)F(\bx_c,\mathcal{T}) \right| \leq C^*\left(h + \delta^2 \right).
	\end{equation*}	
\end{lemma}
\begin{proof}
	We proceed by showing that our differencing stencil reduces to the central difference scheme with additional error sources to consider. Assemble the finite differences, and replace each term with its separated counterparts:
	\begin{eqnarray*}
	\left(\be_k \cdot \nabla_{\delta,h} \right)F(\bx_c,\mathcal{T}) & = & 	\frac{1}{2\mu_T \delta}  
	\left( 
	\left(F_{2k}^{\be_k^+} - \bar{F}_{2k}^{c} \right)
	- \left(F_{2k+2}^{\be_k^-} - \bar{F}_{2k+2}^{c}\right)
	\right)  \\
	& = & \frac{1}{2\mu_T \delta}  
	\left(
	\left(T_{2k}f^{\be_k^+} - \bar{T}_{2k}f^c \right) - \left(T_{2k+2}f^{\be_k^-} - \bar{T}_{2k+2}f^c \right) 
	\right).
	\end{eqnarray*}
	
	\noindent Add and subtract $T_{2k}f^c$ and $T_{2k+2}f^c$ to the right-hand side. When we combine like terms we see
	\begin{eqnarray*}\label{third_eq}
	\text{RHS}  & = &\frac{1}{2\mu_T \delta} T_{2k}\left(f^{\be^+} - f^c\right) \nonumber \\
	& + &\frac{1}{2\mu_T \delta} \left( \left(T_{2k} - \bar{T}_{2k}\right) - \left(T_{2k+2} - \bar{T}_{2k+2} \right) \right)f^c  \nonumber \\ 
	& - &\frac{1}{2\mu_T \delta} T_{2k+2}\left(f^{\be^-} - f^c\right).
	\end{eqnarray*}
	Next, we add and subtract $T_{2k+2}\left(f^{\be^+} - f^c \right)$ and $\mu_T \left(f^{\be^+}-f^{\be^-}\right)$, then collect like-terms such that
	\begin{eqnarray}\label{fourth_eq}
	\text{RHS} & = & \frac{f^{\be^+}-f^{\be^-}}{2\delta}  \nonumber \\
	& + & \frac{1}{2\mu_T \delta} \left(T_{2k}-T_{2k+2}\right)\left(f^{\be^+} - f^c\right) \nonumber \\
	& + &\frac{1}{2\mu_T \delta} \left( \left(T_{2k} - \bar{T}_{2k}\right) - \left(T_{2k+2} - \bar{T}_{2k+2} \right) \right)f^c  \nonumber \\ 
	& + &\frac{1}{2\mu_T \delta} \left(T_{2k+2} - \mu_T \right)\left(f^{\be^+}-f^{\be^-}\right) .
	\end{eqnarray}
	The RHS in \eqref{fourth_eq} includes four terms, for which the first is the central-difference term, and the remaining three account for the error introduced by the dynamic amplitude of $F$. We now estimate each term independently in \eqref{eq22} - \eqref{overEst}. The first of the three remaining terms can be directly estimated by Taylor's theorem in space and time. For the two that remain, we note that there exists a $t^* \in \mathcal{T}$ such that $T(t^*) = \mu_T$. We further note that $|t_{2k+2} - t^*| < (4N+1)h$, and use Taylor's theorem. As a result, there exist $\xi_1, \xi_2, \xi_3 \in \mathcal{T}$ such that
	\begin{align}
	\frac{1}{2\mu_T} \left|\left(T_{2k}-T_{2k+2}   \right) \left(f^{\be_k^+} - f^c \right) \right|	 \leq  	\frac{h}{2\mu_T} \left| T'(\xi_1) \right| \left| \frac{f^{\be_k^+} - f^c}{\delta} \right|, \label{eq22} \\
	\frac{1}{2\mu_T \delta} \left| \left( \left(  T_{2k} - \bar{T}_{2k}\right) - \left(T_{2k+2} - \bar{T}_{2k+2} \right) \right)f^c \right|  \leq  \frac{h^4}{\mu_T \delta } \left| T^{(4)}(\xi_2) \right| \left| f^c \right|, \\
	\frac{1}{2\mu_T \delta}  \left| \left(T_{2k+2} - \mu_T\right)\left(f^{\be_k^+} - f^{\be_k^-}\right) \right|   \leq  \frac{h(4N+1)}{\mu_T} \left| T'(\xi_3) \right| \left| \frac{f^{\be_k^+} - f^{\be_k^-}}{2\delta}  \right|. \label{overEst}
	\end{align}
	The statement now follows directly from Taylor's theorem and the triangle inequality.
\end{proof}

What follows next is a proof of Theorem \ref{NoiseFreeThm}, which extends the above result to higher dimensions.

\begin{proof}[Proof of Theorem \ref{NoiseFreeThm}]
    Begin by noting
    \begin{eqnarray*}
        \left|\left|  \nabla f(\bx_c) - \nabla_{\delta,h} F(\bx_c,\mathcal{T})\right|\right| &=& \left|\left|  \left(\nabla f(\bx_c) - \nabla_\delta f(\bx_c) \right) +\left( \nabla_\delta f(\bx_c) - \nabla_{\delta,h} F(\bx_c,\mathcal{T}) \right)\right|\right| \\
        &\leq& \left|\left|  \nabla f(\bx_c) - \nabla_\delta f(\bx_c) \right|\right| +\left|\left| \nabla_\delta f(\bx_c) - \nabla_{\delta,h} F(\bx_c,\mathcal{T}) \right|\right|.
    \end{eqnarray*}
	An $\mathcal{O}(\delta^2)$ estimate for the first term of the RHS follows directly from Taylor's theorem. For the second, let $\hat{\be}_i$ denote the $i$-th basis vector for $\mathbb{R}^n$. The $i^{th}$ term of the gradient estimate is written
	\begin{equation*}
	\hat{\be}_i^T \cdot \left( \nabla_{\delta,h}F(\bx_c,\mathcal{T})\right) = \sum_{k=1}^{N} \frac{ \hat{\be}_i^T \cdot \be_k }{2 \mu_T N \delta} \left[\left(F_{2k}^{\be_k^+} - \bar{F}_{2k}^{c} \right) - \left(F_{2k+2}^{\be_k^-} - \bar{F}_{2k+2}^{c}\right)\right].
	\end{equation*}
	The full result follows directly from the triangle inequality and Lemma \ref{mainLemma}.
\end{proof}

Before proceeding to the main result, we provide the following estimate on the expectation of n-dimensional Gaussian random variables. We make use of this result in the proof that follows. 
\begin{lemma}\label{expect_lemma} Let $\bq$ be an $n-$dimensional Gaussian random variable with covariance $\Sigma$. Then,
\begin{equation*}
    E\left[||\bq||\right] \leq 4 \sqrt{n||\Sigma||_2}
\end{equation*}
and for all $p \in (0,1)$
\begin{equation*}
    \mathcal{P}\left[ ||\bq||  \leq 4 \sqrt{n||\Sigma||_2} + 2 \sqrt{\log(1/p)||\Sigma||_2}\right] \geq 1-p.
\end{equation*}
\end{lemma}
Proof of this theorem can be found in \cite{vershynin2018high,wainwright2019high}.

Lastly, the results above can be combined to illicit a a proof of Theorem \ref{NoisyThm}, which is provided below.

\begin{proof}[Proof of Theorem \ref{NoisyThm}]
Recall $\hat{\be}_i$ denotes the $i$-th standard basis vector of $\mathbb{R}^n$, and $\be_k$ denotes a unit vector co-linear with the $k$-th antipodal pair. Let ${\bf \varsigma}$ denote the summation of independent noise terms, i.e.,
\begin{equation*}
    \varsigma_k = \frac{\varsigma_{1,k}}{2} + \varsigma_{2,k} + \varsigma_{3,k} + \frac{\varsigma_{4,k}}{2}.
\end{equation*}
We see that ${\bf \varsigma} \sim \mathcal{N}(0,3\sigma^2)$, where $\sigma^2$ is the given variance of the additive noise. We can then write
\begin{equation*}
    \varepsilon_i = \sum_{k=1}^N \frac{\hat{\be}_i^T \cdot \be_k}{2 \mu_T \delta N} \varsigma_k,
\end{equation*}
which can be written
\begin{equation*}
    {\bf \varepsilon} = {\bf A \varsigma}.
\end{equation*}
Thus, $\bf \varepsilon$ is distributed $\mathcal{N}\left(0,3\sigma^2 {\bf A A}^T \right).$ Next, we define $\Sigma = 3\sigma {\bf A A}^T$, from which we compute
\begin{eqnarray*}
\Sigma_{i,j} &=& 3\sigma^2({\bf A A}^T)_{i,j} \\
& = & \frac{3 \sigma^2}{4 \mu_T^2 \delta^2 N^2} \sum_{k=1}^N (\hat{\be}_i^T \cdot \be_k)^T(\hat{\be}_j^T \cdot \be_k) \\
& = & \frac{3 \sigma^2}{4 \mu_T^2 \delta^2 N^2} \left( \sum_{k=1}^N ||\be_k||^2 \right) \left(\hat{\be}_i^T \cdot \hat{\be}_j\right) 
\end{eqnarray*}
and therefore
\begin{equation}\label{sigma_estimate}
  ||\Sigma||_2 = \frac{3\sigma^2}{4\mu_T^2 \delta^2 N}.
\end{equation}
The statement now follows directly from Lemma \ref{expect_lemma}.
\end{proof}

\end{document}